\documentclass[a4paper,UKenglish, cleveref, autoref, thm-restate, nosubfigcap, nolineno]{socg-lipics-v2021}

\usepackage{cite}
\usepackage{booktabs}
\usepackage{nicematrix}
\usepackage{pdfpages}
\usepackage{tikz}
\usepackage{amsmath}

\newtheorem*{problem*}{Problem}

\renewcommand{\epsilon}{\varepsilon}

\title{Improved Lower Bound on the Number of Pseudoline Arrangements\footnote{
This manuscript was accepted at SoCG'24
and will be merged
with Fernando Cortés Kühnast, Stefan Felsner and Manfred Scheucher's manuscript
``An Improved Lower Bound on the Number of Pseudoline Arrangements'' for the proceedings}}

\author{Justin Dallant}{Université libre de Bruxelles, Belgium} {Justin.Dallant@ulb.be}{https://orcid.org/0000-0001-5539-9037}{Supported by the French Community of Belgium via the funding of a FRIA grant.}

\authorrunning{J. Dallant}

\Copyright{Justin Dallant}



\acknowledgements{}%optional

\nolinenumbers

\hideLIPIcs

\begin{document}

\maketitle

\begin{abstract}
    We show that for large enough $n$, the number of non-isomorphic pseudoline arrangements of order $n$ is greater than $2^{c\cdot n^2}$ for some constant $c > 0.2604$, improving the previous best bound of $c>0.2083$ by Dumitrescu and Mandal (2020). Arrangements of pseudolines (and in particular arrangements of lines) are important objects appearing in many forms in discrete and computational geometry. They have strong ties for example with oriented matroids, sorting networks and point configurations. Let $B_n$ be the number of non-isomorphic pseudoline arrangements of order $n$ and let $b_n := \log_2(B_n)$. The problem of estimating $b_n$ dates back to Knuth, who conjectured that $b_n \leq 0.5n^2 + o(n^2)$ and derived the first bounds $n^2/6-O(n) \leq b_n \leq 0.7924(n^2+n)$. Both the upper and the lower bound have been improved a couple of times since. For the upper bound, it was first improved to $b_n < 0.6988n^2$ (Felsner, 1997), then $b_n < 0.6571 n^2$ by Felsner and Valtr (2011), for large enough $n$. 
    In the same paper, Felsner and Valtr improved the constant in the lower bound to $c> 0.1887$, which was subsequently improved by Dumitrescu and Mandal to $c>0.2083$.
    Our new bound is based on a construction which starts with one of the constructions of Dumitrescu and Mandal and breaks it into constant sized pieces. We then use software to compute the contribution of each piece to the overall number of pseudoline arrangements. This method adds a lot of flexibility to the construction and thus offers many avenues for future tweaks and improvements which could lead to further tightening of the lower bound.
\end{abstract}

\section{Introduction}

\begin{figure}[t]
    \centering
    \includegraphics[width=0.4\textwidth]{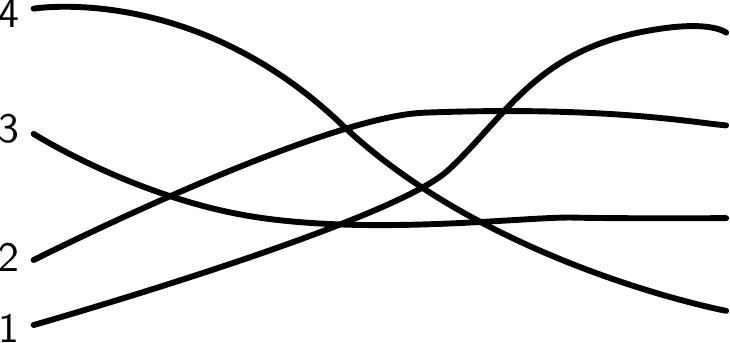}
    \caption{A pseudoline arrangement of order $4$ with its natural labeling.}
    \label{fig:labeling}
\end{figure}

A pseudoline arrangement of order $n$ is a collection of $n$ unbounded $x$-monotone simple curves in the Euclidean plane such that every pair of curves intersects exactly once, and crosses at their intersection. It is said to be simple if no three curves meet at a point. Every pseudoline arrangement induces a natural labeling on the curves composing it, given by the order in which they appear, from bottom to top, to the left of the first intersection (see Figure \ref{fig:labeling}). We say two pseudoline arrangements are isomorphic if they can be mapped to each other by a homeomorphism of the plane preserving this natural labeling\footnote{There are notions of pseudoline arrangements and of their isomorphism in the literature which differ from the one we use here (e.g. working in the projective plane, having explicitly labeled curves or not requiring the mappings to preserve the labelings). In terms of counting the number of non-isomorphic pseudoline arrangements of order $n$, these differences only affect lower order factors and have no bearing on the results of this paper.} (or more intuitively, they are ``combinatorially identical'').

Isomorphism classes of (simple) pseudoline arrangements (or variations thereof) appear in many different forms in discrete geometry and beyond. Among those are configurations of points, wiring diagrams, zonotopal tilings, reorientation classes of oriented matroids of rank 3 (see e.g. \cite{Handbook} and \cite{FelsnerBook2004} for more general treatments of the topic).

\begin{table}[ht]
    \centering
    \begin{tabular}{c|r}
        $n$ & $B_n$  \\ \hline\hline
        $1$ & $1$ \\ \hline
        $2$ & $1$ \\ \hline
        $3$ & $2$ \\ \hline
        $4$ & $8$ \\ \hline
        $5$ & $62$ \\ \hline
        $6$ & $908$ \\ \hline
        $7$ & $24\ 698$ \\ \hline
        $8$ & $1\ 232\ 944$ \\ \hline
        $9$ & $112\ 018\ 190$ \\ \hline
        $10$ & $18\ 410\ 581\ 880$ \\ \hline
        $11$ & $5\ 449\ 192\ 389\ 984$ \\ \hline
        $12$ & $2\ 894\ 710\ 651\ 370\ 536$ \\ \hline
        $13$ & $2\ 752\ 596\ 959\ 306\ 389\ 652$ \\ \hline
        $14$ & $4\ 675\ 651\ 520\ 558\ 571\ 537\ 540$ \\ \hline
        $15$ & $14\ 163\ 808\ 995\ 580\ 022\ 218\ 786\ 390$ \\ \hline
        $16$ & $76\ 413\ 073\ 725\ 772\ 593\ 230\ 461\ 936\ 736$ \\ \hline
    \end{tabular}
    \caption{The known values of $B_n$. The values for $n\leq 9$ are from \cite[p.\ 35]{Knuth1992}. The values for $n=10,11,12$ are from \cite{Felsner97}, \cite{Yamanaka2010} and \cite{oeis} respectively. The value for $n=16$ is from \cite{Rote2023}.}
    \label{tab:Bn}
\end{table}

Let $B_n$ be the number of non-isomorphic simple pseudoline arrangements of order $n$, and let $b_n := \log_2 B_n$. The values of $B_n$ for $1\leq n \leq 16$ are known exactly and can be found in Table \ref{tab:Bn}. Here, we are interested in bounding $b_n$ as a function of $n$. In his book \cite{Knuth1992}, Knuth conjectured that $b_n \leq 0.5n^2 + o(n^2)$ and derived the first bounds $n^2/6-O(n) \leq b_n \leq 0.7924n^2+O(n)$. The upper bound was derived by showing that the number of cutpaths in a wiring diagram of order $n$ (which corresponds to the number of combinatorially distinct ways to insert a new curve in a pseudoline arrangement of order $n$) is at most $3^n$, yielding the bound $B_n \leq 3^{\binom{n}{2}} \leq 2^{0.7924n^2+O(n)}$. The lower bound was obtained by a recursive construction, yielding the recurrence $B_n \geq 2^{n^2/8-n/4}B_{n/2}.$ A more geometric (albeit with smaller constant) recursive construction was later given by Matou\v{s}ek in his book \cite{Matousek2002}. We review this latter construction below. 

The upper bound was subsequently improved (for large enough $n$) to $b_n < 0.6988n^2$ by Felsner \cite{Felsner97} through a clever encoding of the arrangement, then to $b_n < 0.6571 n^2$ by Felsner and Valtr \cite{FelsnerValtr2011}, through an improved bound on the number of cutpaths.  In the same paper, the lower bound was improved to $b_n \geq 0.1887n^2$ by a recursive construction making use of the number of rhombic tilings of a centrally symmetric hexagon. Finally, through constructions which can be seen as generalizations of that of Matou\v{s}ek, Dumitrescu and Mandal \cite{Dumitrescu2020} improved the lower bound to $b_n \geq 0.2083n^2$.

The present work builds on one of the constructions of Dumitrescu and Mandal (which we briefly review in the beginning of Section \ref{sec:construction}), by making the choices throughout less local, and with the use of computer software to count the number of choices possible in each ``piece'' of the construction. Our construction yields $b_n \geq 0.2604n^2$, for large enough $n$.

\paragraph*{Matou\v{s}ek's construction.}
\begin{figure}
    \centering
    \includegraphics[width=0.6\linewidth]{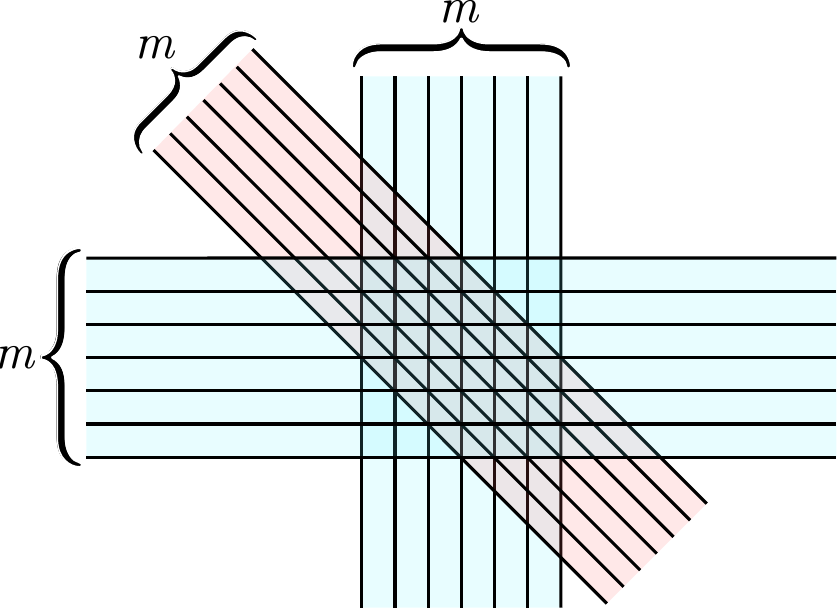}
    \caption{Matou\v{s}ek's lower bound construction.}
    \label{fig:matousek_bound}
\end{figure}

Matou\v{s}ek \cite[Sec. 6.2]{Matousek2002} gave a simple construction yielding a lower bound of $2^{n^2/8}$ on the number of pseudoline arrangements of order $n$. One can describe this construction by starting with $2$ bundles of lines, each consisting of $m = n/3$ horizontal (resp. vertical) lines with unit distance between consecutive lines. This defines a regular grid of $m^2$ points. We then insert a third bundle of $m$ lines with slope $-1$, which together pass through $3(m^2+1)/4$ points of the grid, assuming $m$ is odd (see Figure \ref{fig:matousek_bound}). For each of these, we have two choices leading to a different arrangement: we can bend the line to pass either above or below the intersection.\footnote{As described, the obtained arrangement does not consist of $x$-monotone curves, but this can be remedied simply by rotating the grid slightly. We choose not to do so for consistency with the later constructions in the paper.} It then remains to bound the number of ways in which the lines within each of the bundles can be arranged with respect to each-other. This leads to the recurrence
\[B_{3m} \geq 2^{3m^2/4} (B_m)^3,\]
which solves to $B_n\geq 2^{n^2/8}$.
\section{Pseudochord arrangements}

\begin{figure}
    \centering
    \includegraphics[width=0.45\linewidth]{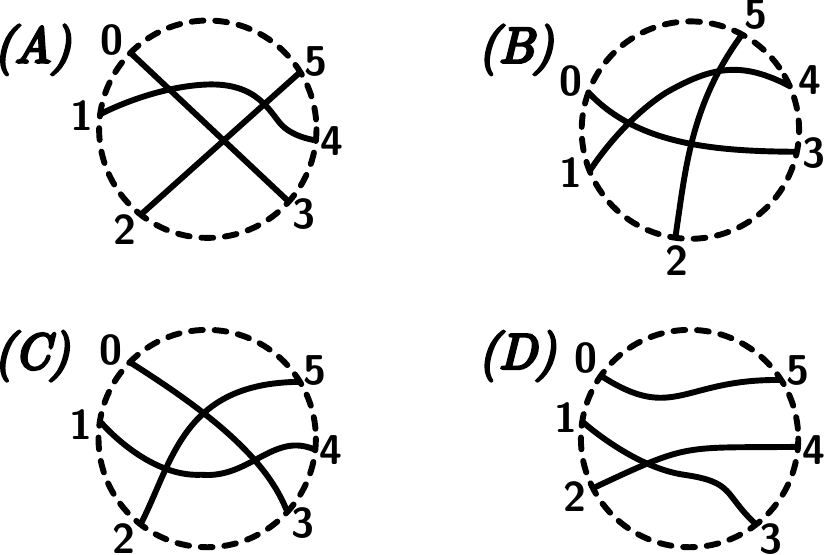}
    \caption{Illustrations $(A)$ and $(B)$ represent two different embeddings of the same pseudochord arrangement, while $(C)$ is a different arrangement with the same matching. Illustration $(D)$ is an embedding of a matching distinct from the three others.}
    \label{fig:diagrams}
\end{figure}

For our construction, we will work with a generalization of pseudoline arrangements, which we call pseudochord arrangements. Note that similar objects have been considered previously, under the name of \emph{weak pseudoline arrangements} \cite{Fraysseix2003, Eppstein2013} and \emph{partial pseudoline arrangements} \cite{FelsnerValtr2011,Rote2023}. 

\begin{definition}
Consider a simple closed curve $C$ in the plane. Consider also $2k$ distinct points on $C$, labeled $0$ to $2k-1$, in counter-clockwise order around $C$. Consider a perfect matching $\{a_1,b_1\}, \{a_2,b_2\}, \ldots \{a_k,b_k\}$ on the $2k$ points. For each pair in the matching, add a simple curve with the corresponding points as endpoints, labeled by these endpoints, with the following restrictions:
\begin{itemize}
    \item the curves lie inside the bounded side of $C$, with the exception of their two endpoints (which lie on $C$);
    \item two curves intersect at most once, and if they do, they cross at the point of intersection.
\end{itemize}

We call this set of curves (including $C$), together with the $2k$ labeled points an embedded pseudochord arrangement of order $k$. We refer to $C$ as the bounding curve, to the $k$ other curves as (pseudo-)chords, and to the $2k$ points as endpoints (of the corresponding chords). If no three chords intersect at a common point, the arrangement is said to be simple.
\end{definition}

We can give a canonical orientation to each chord in an embedded pseudochord arrangement based on the labels of the endpoints: for a chord connecting two endpoints labeled $a$ and $b$, with $a<b$, we orient the chord from $a$ to $b$. This allows us to unambiguously distinguish between the two regions inside the bounding curve $C$ on both sides of the chord, which we call ``above'' and ``below'' the chord. 

\begin{definition}
    We call chirotope of a pseudochord arrangement the mapping which to each triplet $(c_1,c_2,c_3)$ of chord labels in a pseudochord arrangement assigns $\bot$ if $c_1$ and $c_2$ do not intersect, $1$ if their intersection is above $c_3$, $-1$ if it is below, and $0$ if the three chords meet at a single point.

    We say that two embedded pseudochord arrangements are isomorphic (or, more bluntly, ``the same'') if they have the same chirotope. We call pseudochord arrangement the equivalence class of an embedded pseudochord arrangement under this relationship.
    
    We call any embedded pseudochord arrangement in this equivalence class an embedding of the pseudochord arrangement.
\end{definition}

We further group pseudochord arrangements by the matchings they define.
\begin{definition}
    Given some pseudochord arrangement $A$, we let $M_A$ denote the matching it defines on its labelled endpoints.
    
    Given some perfect matching $M$ on $\{0,1,\ldots 2k-1\}$, We let $arr(M)$ denote the set of simple\footnote{Throughout the paper, we will only be concerned with counting simple arrangements, and will thus omit this qualifier in most places.} pseudochord arrangements $A$ with $M_A = M$.  We call embedding of $M$ any embedding of a pseudochord arrangement in $arr(M)$.
\end{definition}

Figure \ref{fig:diagrams} illustrates these definitions.
Note that for two matchings of order $k$ which differ only by a relabeling of the endpoints which shifts them cyclically (i.e. a mapping of the labels of the form $s_t:a\mapsto a+t\ [mod\ 2k]$) there is a natural one-to-one correspondence between their pseudochord arrangements. In particular, when we are only interested in counting the number of pseudochord arrangements of a matching, we may instead choose any arbitrary such shift of the matching, or not even specify which of these we consider.

\paragraph*{A family of relevant matchings.}
Let us define and give notations for some matchings which are particularly relevant to the present paper.

\begin{definition}
    Let $k_1, k_2, \ldots k_r$ be strictly positive integers and let $k$ be their sum. We will use the notation $(k_1, k_2, \ldots, k_r)$-matching to denote the matching corresponding to $r$ groups of chords of size $k_1, \ldots, k_r$ respectively, such that two chords cross if and only if they belong to different groups (see Figure \ref{fig:3-2-4-diagram} for an illustration).
    If $k_1=k_2=\ldots=k_r = k/r$, we will further shorten the notation to $(k/r)_r$-matching.
\end{definition}

\begin{figure}
    \centering
    \includegraphics[width=0.2\linewidth]{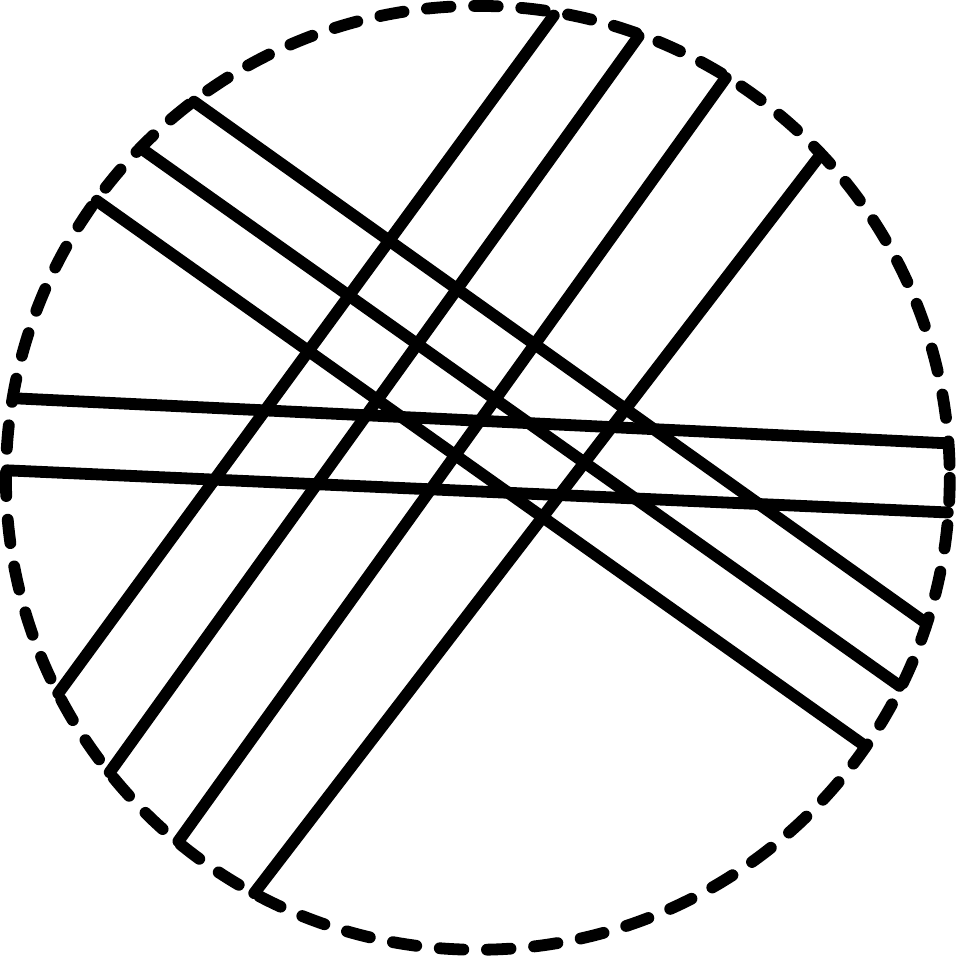}
    \caption{Embedding of a $(3,2,4)$-matching.}
    \label{fig:3-2-4-diagram}
\end{figure}

In particular, an embedding of a $(1)_n$-matching has $n$ chords, all pairwise intersecting. Moreover, its pseudochord arrangements are in one-to-one correspondence with the pseudoline arrangements of order $n$.

\paragraph*{Some basic facts about pseudochord arrangements}

The following definition and propositions will form the basis of our construction.

\begin{definition}
    Consider an embedding $E$ of a matching, with bounding curve $C$. Let $C'$ be a simple closed curve contained in the union of $C$ and the bounded side of $C$. If $C$ does not intersect any point where two chords of $E$ cross, then it naturally defines an embedding of a matching, up to cyclical shift of the endpoint labels, which we call a subembedding of $E$. We say that two such subembeddings are independent if the bounded sides of their respective bounding curves are disjoint.
\end{definition}

\begin{proposition}\label{prop:sub_diagram}
Let $E$ be an embedding of a matching, and let $\{E_1,E_2,\ldots,E_r\}$ be a set of pairwise independent subembeddings of $E$. Then,
\[|arr(M_E)| \geq \prod_{i=1}^r |arr(M_{E_i})|.\]
\end{proposition}
\begin{proof}
    By perturbing the chords slightly, we can assume without loss of generality that the initial embedded matching $E$ is simple.
    For every choice of a simple arrangement in each of the subembeddings, we can choose an embedding with the same bounding curve and same chord endpoints as the starting subembedding. Because in every simple arrangement of a matching the same pairs of chords cross, the overall embedding obtained by locally replacing each subembedding with the chosen arrangement defines a valid simple arrangement of $E$. Moreover, if two overall arrangements differ in at least one of their subembeddings then they are non-isomorphic, as their chirotopes differ for at least one chord triplet.
\end{proof}

The same recursive procedure as in Matou\v{s}ek's recursive construction described above gives the following.
\begin{proposition}\label{prop:recursion}
    Let $r > 0$ be a positive integer. For any positive integer $m$, let $D_r^m$ denote a $(m)_r$-matching. If $c$ is a constant such that $\log_2 |arr(D_r^m)| \geq {c\cdot m^2-O(m)}$, then $b_n \geq {\frac{c}{r(r-1)}n^2 -O(n\log n)}$.
\end{proposition}

\section{Counting pseudochord arrangements of small matchings}\label{section:counting}

In this section we describe a way to compute the number of pseudochord arrangements for relatively small matchings. The general idea follows the straightforward incremental approach of adding chords one after the other and generating all possible resulting arrangements. There are however some tricks which can greatly speed up the computation for some specific matchings. We will keep the descriptions somewhat succinct, as this is not the main focus of the paper.

\paragraph*{General approach}

The general approach we use is straightforward and incremental. Let $c_1, c_2, \ldots, c_k$ denote the chords in a certain order. We incrementally construct embeddings for all pseudochord arrangements on the first $i$ chords (represented as doubly-connected edge lists, or DCELs) for $1\leq i < k$, by taking all constructed embeddings on the first $i-1$ chords and for each of them inserting chord $c_i$ in all combinatorially distinct valid ways. This amounts, for each of them, to generating all possible paths in a certain directed acyclic graph (or DAG for short). Once we have all embeddings for the first $k-1$ chords, we count the number of ways to insert the chord $c_k$ in each of them. This amounts to counting the number of paths in a certain DAG (where previously, we were generating them instead of simply counting them).

Because counting the number of paths in a DAG can be done much faster in general than generating all paths, it makes sense for us to insert the chord with the largest contribution to the total number of pseudochord arrangements last. We next discuss how we can push this heuristic a bit further for some types of matchings.

\paragraph*{Independent chords}

\begin{figure}
    \centering
    \includegraphics[width=0.3\linewidth]{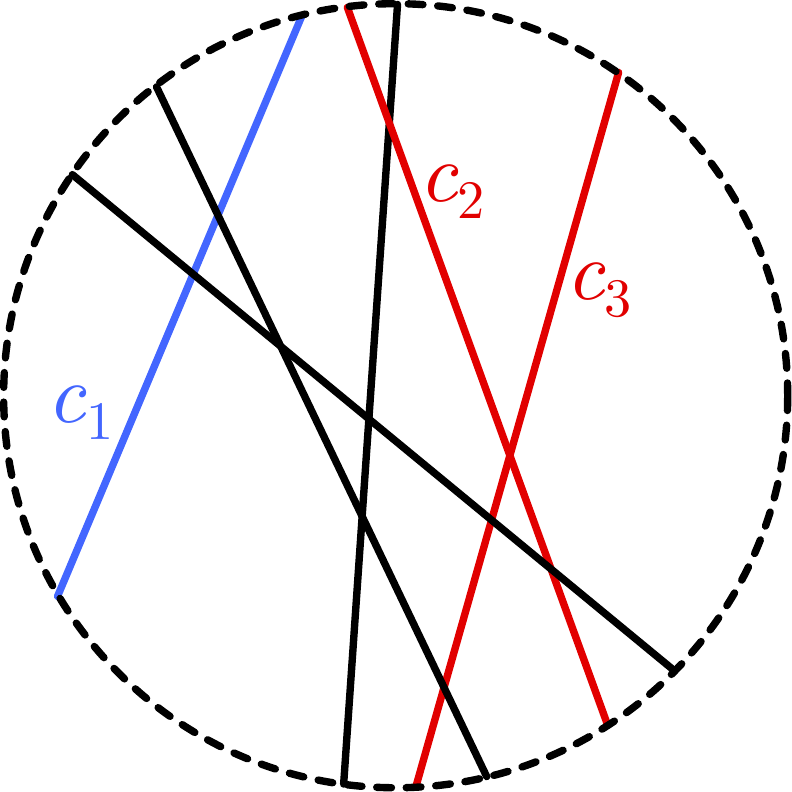}
    \caption{Example of an embedded matching with independent chords. The blue chord $c_1$ is independent from both $c_2$ and $c_3$.}
    \label{fig:independent_chords}
\end{figure}

Consider the embedding of a matching represented in Figure \ref{fig:independent_chords}. A bit of thought should reveal that once we have inserted the black chords, the number of combinatorially distinct ways to insert the blue chord $c_1$ (while still representing the same matching) does not depend on how the red chords $c_2$ and $c_3$ are inserted (or indeed, on their existence at all), and reciprocally the number of ways to insert the red chord $c_2$ and $c_3$ does not depend on how the blue chord $c_1$ is inserted. In this sense, we can think of the red chords as being independent of the blue chord. We formalize this notion as follows.\footnote{Note that the formal definition we give only constitutes a sufficient condition for chords to be independent in the intuitive sense discussed before, but this formal definition is what our algorithm works with.}
\begin{definition}
    Let $M$ be an embedded matching. We say that two chords $c$ and $c'$ of $M$ are independent if at least one of the following conditions hold:
    \begin{itemize}
        \item there exists a chord $k$ of $M$ such that $c$ and $c'$ lie on different sides of $k$;
        \item $c$ and $c'$ do not intersect, and no two other chords $k$ and $k'$ which both intersect $c$ and $c'$ intersect each other;
        \item no other chord intersects both $c$ and $c'$. 
    \end{itemize}
\end{definition}

Our previous discussion suggests the following algorithm: Given an embedded matching $M$, partition its chords into three sets $R$, $G$ and $B$, such that all chords in $R$ are independent from all chords in $B$. Then, for every possible way to insert the chords of $G$, independently and recursively count the number of ways to insert the chords of $R$ and those of $B$, and multiply these counts together. The sum over all arrangements of the chords in $G$ gives the final count.
\begin{figure}
    \centering
    \includegraphics[width=0.6\linewidth]{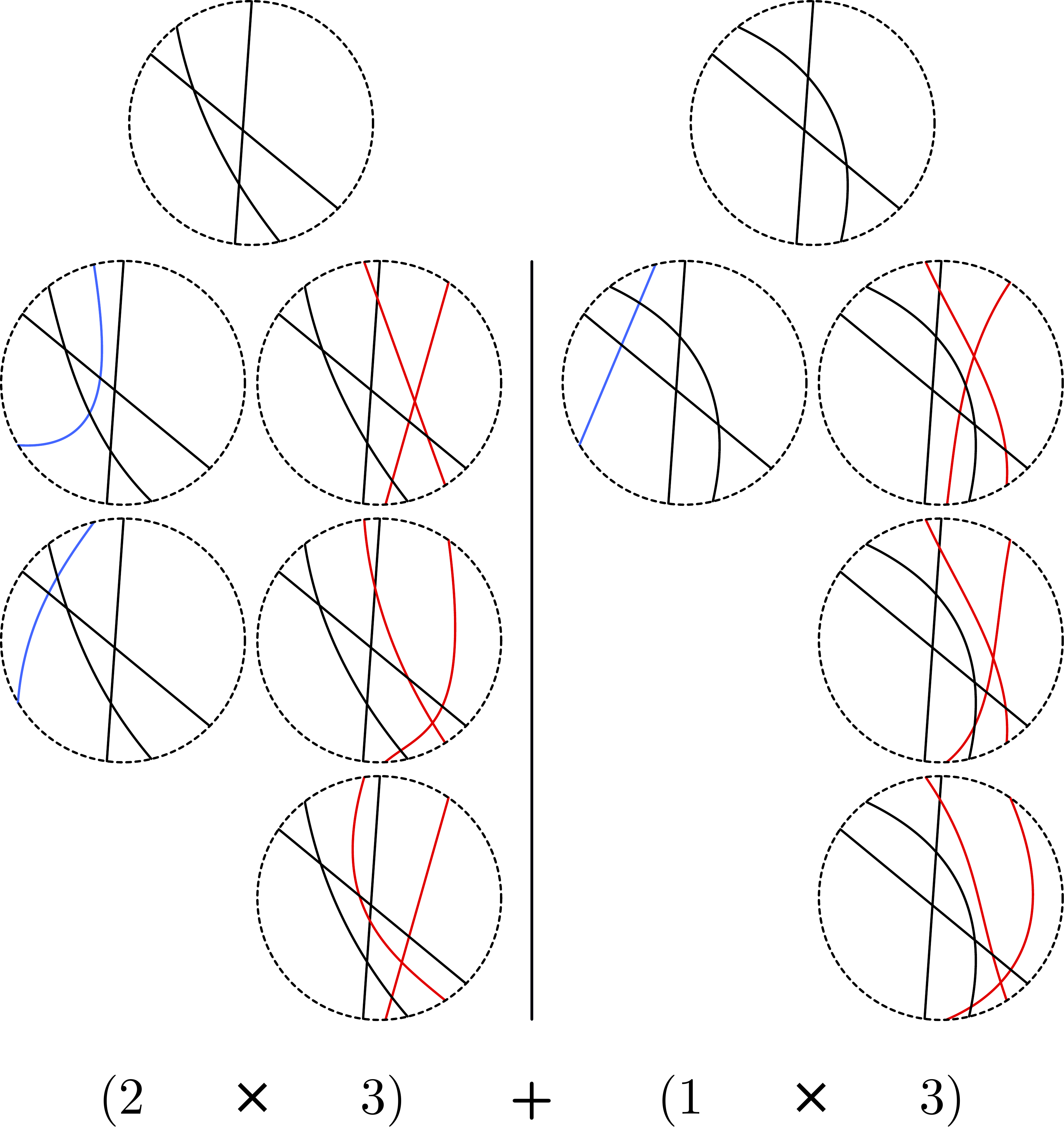}
    \caption{Application of the independent chord approach on the example of Figure \ref{fig:independent_chords}.}
    \label{fig:independent_chords_count}
\end{figure}
This procedure, applied to the example of Figure \ref{fig:independent_chords}, is illustrated in Figure \ref{fig:independent_chords_count}. While the advantage of this method might not be immediately apparent on such a small example, it should be clear that this can offer great performance improvements on larger matchings if there are enough independent chords.

The question remains on how to perform the partition. Unfortunately, we do not know how to obtain the partition which will most speed up the computation, so in our implementation we have relied purely on heuristics for this. In short, our approach is the following:
\begin{itemize}
    \item Compute all pairs of independent chords.
    \item For each chord $c_i$, estimate its contribution to the total number of line arrangements by sampling arrangements of the other chords and counting the number of ways to insert $c_i$ in these sampled arrangements. This defines a weight $w_i$ for the chord.
    \item Generate many valid partitions of the chords into sets $R,G,B$. The weight of a set is the product of the weights of its chords. Choose the partition which maximizes the minimum of the weights of $R$ and $B$ (this captures the intuition that ``heavy'' chords should be inserted last, and that we want to balance the weight as much as possible across the two subproblems).
    \item Recurse on $R$ and $B$.
\end{itemize}
In practice, this simple heuristic gave speed-ups of multiple orders of magnitude on some of the matchings we considered (compared to not exploiting independent chords at all), although more work on this could probably offer further improvements.

\section{Warm-up construction}
\begin{figure}
    \centering
    \includegraphics[width=0.4\linewidth]{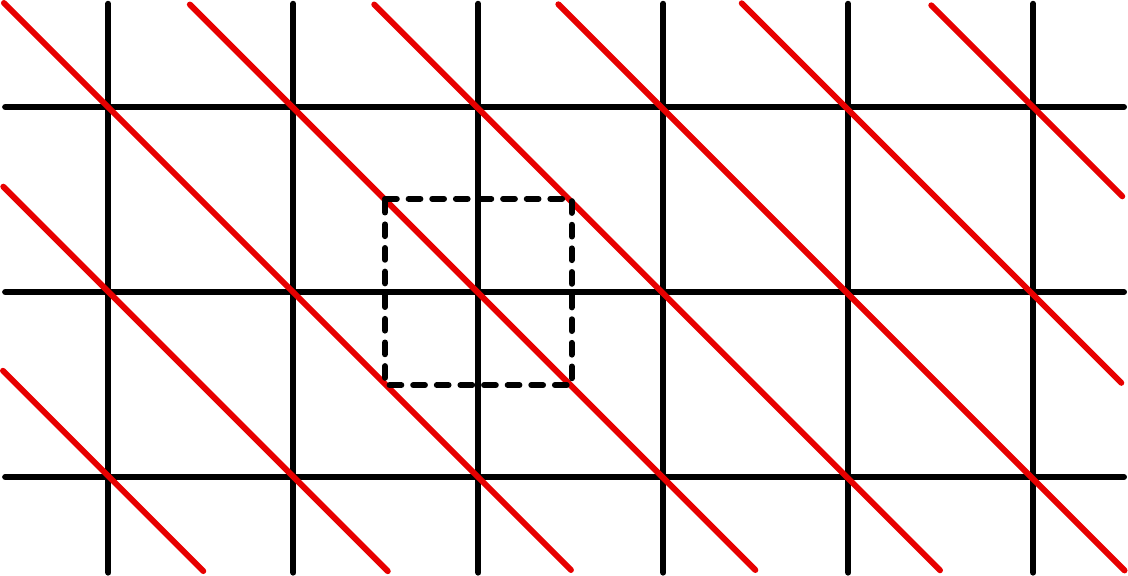}
    \caption{Repeated pattern in the intersection of the three slabs of Matou\v{s}ek's construction. The dotted line represents the boundary of a basic tile of this pattern.}
    \label{fig:matousek-tiles}
\end{figure}
We illustrate our method by slightly improving the lower bound obtained by Matou\v{s}ek's construction. 
Consider the previously mentioned construction (Figure \ref{fig:matousek_bound}) viewed as an embedding of an $(m)_3$-matching. We call the region between the two extremal lines of the same slope a slab. We focus on the region in the intersection of all three slabs, which has area $3m^2/4$. Notice that, ignoring the borders, the lines form a repeating pattern inside that region whose basic tile is an axis-aligned unit square, as illustrated in Figure \ref{fig:matousek-tiles}. These unit squares define a set of identical pairwise independent subembeddings of three pairwise intersecting chords. Each of these subembeddings has $2$ associated pseudochord arrangements. By Proposition \ref{prop:sub_diagram}, this implies that a $(m)_3$-matching has at least $2^{3m^2/4-O(m)}$ pseudochord arrangements. Applying proposition \ref{prop:recursion}, we get a lower bound of $2^{n^2/8-O(n\log n)}$ on the number of pseudoline arrangements of order $n$.

\begin{figure}
    \centering
    \includegraphics[width=0.8\linewidth]{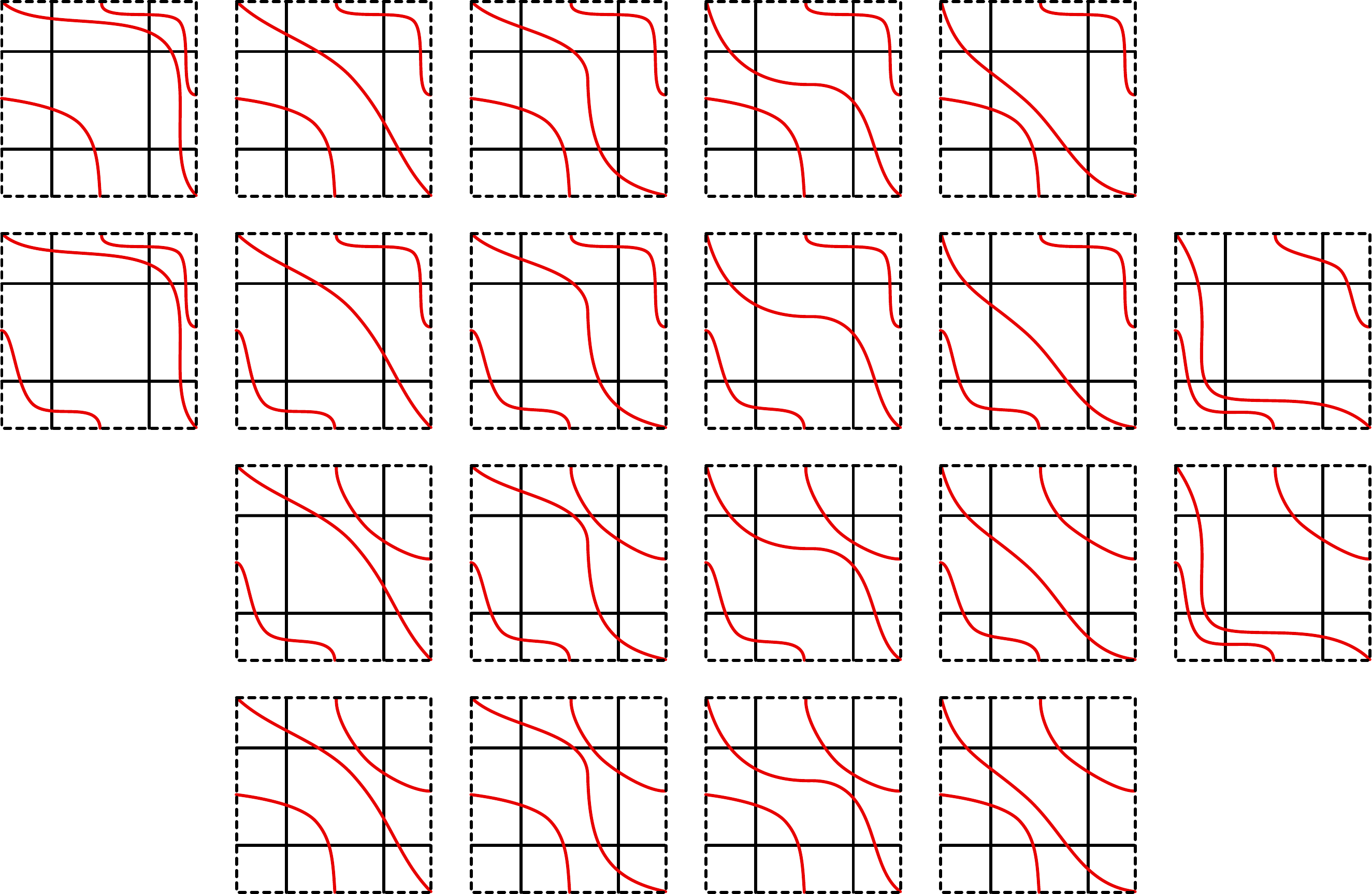}
    \caption{The $20$ arrangements of the subembedding used in our warm-up improvement of Matou\v{s}ek's construction.}
    \label{fig:matousek-diagrams}
\end{figure}
Up to this point, we have only been describing Matou\v{s}ek's construction in different terms. Now comes the improvement: instead of partitioning the area in the intersection of slabs into unit squares, we partition\footnote{We use the word ``partition'' in a loose sense here, as we are ignoring areas near the border of the shape to partition.} it into squares of side-length $2$. The number of such squares inside this area is $3m^2/16-O(m)$ (io
roughly $4$ times fewer than for unit squares, up to border effects). What is the number of pseudochord arrangements for each of the associated matchings? If this number was $2^4 = 16$, we would recover roughly the previous bound again. But it turns out that this number is $20$ (see Figure \ref{fig:matousek-diagrams} which illustrates them all). Thus, by Proposition \ref{prop:sub_diagram}, this implies that a $(m)_3$-matching has at least $20^{3m^2/16-O(m)} = 2^{3\log_2(20) m^2/16-O(m)}$ pseudochord arrangements. Applying proposition \ref{prop:recursion}, we get a lower bound of $2^{\log_2(20) n^2/32-O(n\log n)}$ on the number of pseudoline arrangements of order $n$, where $\log_2(20)/32 > 0.135$ (to be compared with $1/8=0.125$ for Matou\v{s}ek's original construction). We would get further improvements by considering larger and larger subembeddings (although with this particular construction, we cannot hope to improve the constant past the constant of $0.1887...$ obtained by Felsner and Valtr \cite{FelsnerValtr2011}). 

\section{The main construction}\label{sec:construction}

Our main result is based on the same principle as the warm-up. Instead of starting with Matou\v{s}ek's construction, we will start with the ``rectangular construction with 12 slopes'' of Dumitrescu and Mandal. We will use many types of constant sized subembeddings, for which we have explicitly computed the number of pseudochord arrangements (mostly using the method outlined in Section \ref{section:counting}). 

\paragraph*{Rectangular construction with 12 slopes}
% \begin{table}[]
%     \centering
%     \begin{tabular}{c|c}
%        Slope  & $y$-intercepts of the extremal lines\\\hline
%        $\pm 1/3$  & $\pm (m-1)/6$\\
%        $\pm 1/2$  & $\pm (m-1)/4$\\
%        $0, \pm 1, \pm 2, \pm 3$  & $\pm (m-1)/2$\\
%     \end{tabular}
%     \caption{The $y$-intercepts for the extremal lines in each bundle.}
%     \label{tab:y_intercepts}
% \end{table}
We start by recalling one of the constructions of Dumitrescu and Mandal \cite{Dumitrescu2020}. The construction is based on twelve bundles of parallel lines with respectives slopes $0$, $\infty$, $\pm 1/3$, $\pm 1/2$, $\pm 1$, $\pm 2$, $\pm 3$. Each bundle consists of $m$ equally spaced lines, where $m$ is odd, which are placed such that that the middle line passes through the origin. The two extremal lines of the bundle of slope $\infty$ are given by the equations $x=\pm (m-1)/2$. 
The other $10$ extremal lines are given by $y = s\cdot x \pm (m-1)/2$ for $s = 0, \pm 1, \pm 2, \pm 3$ and $y = s\cdot x \pm s(m-1)/2$ for $s =\pm 1/2, \pm 1/3$.

% the $y$-intercepts for the other slopes are listed in Table \ref{tab:y_intercepts}.

From here on, Dumitrescu and Mandal count the number of points where $i$ lines meet, for $3\leq i \leq 12$, and argue that for each of these points there are $B_i$ possible choices for the ``local arrangement'' around the intersection point. Multiplying all of these choices together gives their bound.

\begin{figure}[t]
    \centering
    \includegraphics[width=0.8\linewidth]{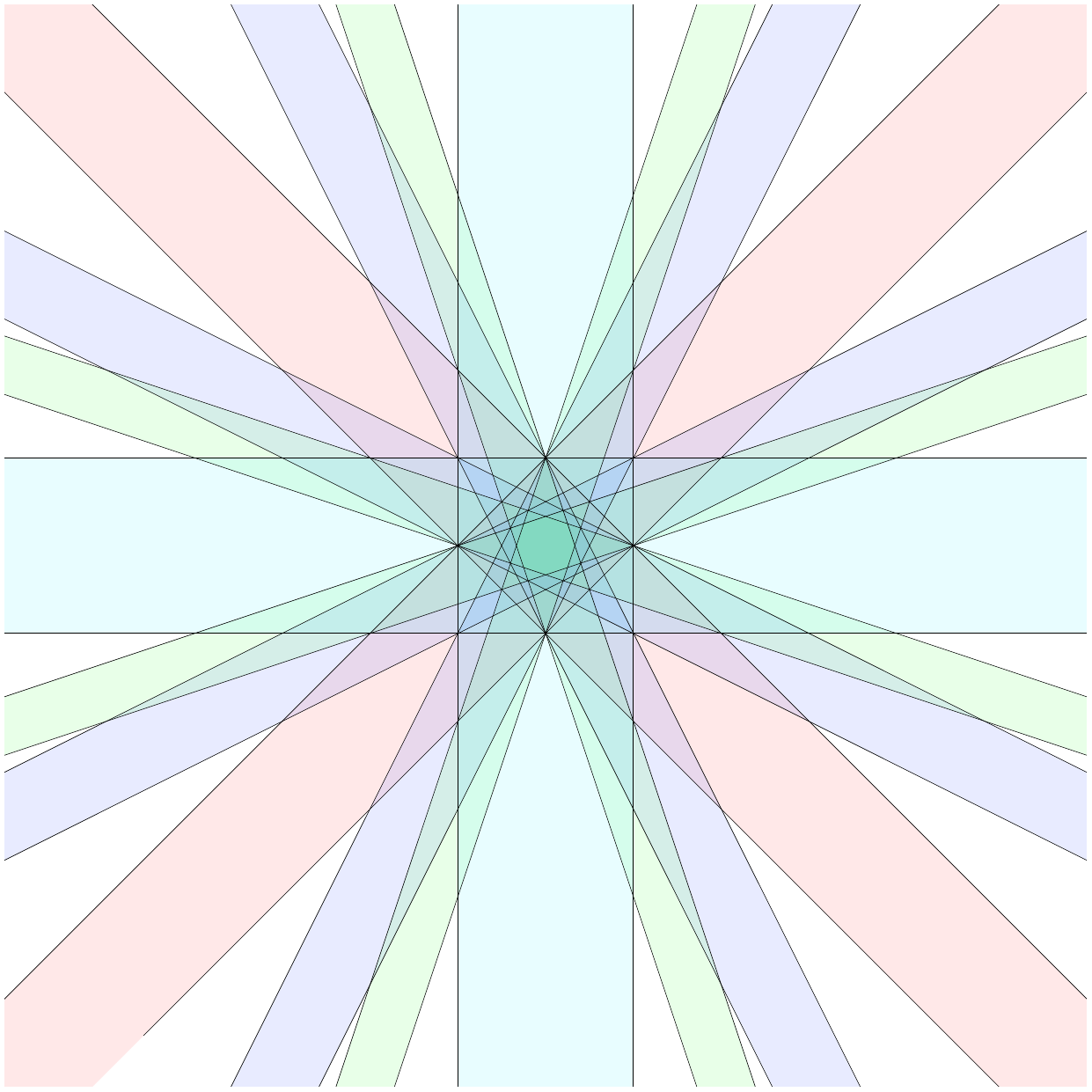}
    \caption{Illustration of the extremal lines and slabs of the construction.}
    \label{fig:rect_12}
\end{figure}

\paragraph*{Our bound}

We view the previous construction of lines as an embedding of a $(12)_m$-matching (by adding a bounding curve) and break it into different types of subembeddings, as we did in the warm-up construction. 

We call the area between the two extremal lines of slope $s$ the slab of slope $s$. Figure \ref{fig:rect_12} illustrates these extremal lines and the corresponding slabs. We group points of the plane depending on which set of slabs they belong to, and we further group them by rotations around the origin and reflexions around the lines of slopes $0$, $\infty$,$-1$,$1$ passing through the origin (for example, points which belong to exactly the slabs with slopes $1$, $2$ and $\infty$ would be grouped together with points belonging to exactly the slabs with slopes $-1$, $-2$ and $\infty$).

This grouping gives rise to $19$ relevant regions, which we will denote by $R_A, R_B, \ldots R_S$. For each of these regions, we will compute a contribution to the overall number of pseudochord arrangements, by choosing a set of subembeddings. Proposition \ref{prop:sub_diagram} will then give us a lower bound on the number of pseudochord arrangements of a $(12)_m$-matching. To compute this contribution we will need to know the area of each of these $19$ regions. We refer the reader to Dumitrescu and Mandal for how these can be computed.\footnote{We note here also that the computation of these areas can be fully automated with the use of standard computer algebra software, thus possibly allowing for computer search of constructions like those of Dumitrescu and Mandal.}

\paragraph*{Region $R_A$}

\begin{figure}[h]
    \centering
    \includegraphics[width=0.35\linewidth]{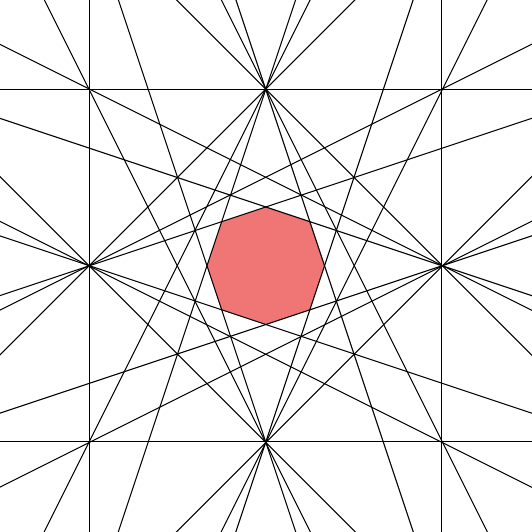}
    \caption{Extremal lines of the construction, with the region $R_A$ higlighted in red.}
    \label{fig:RA}
\end{figure}

\begin{figure}[h]
    \centering
    \includegraphics[width=0.95\linewidth]{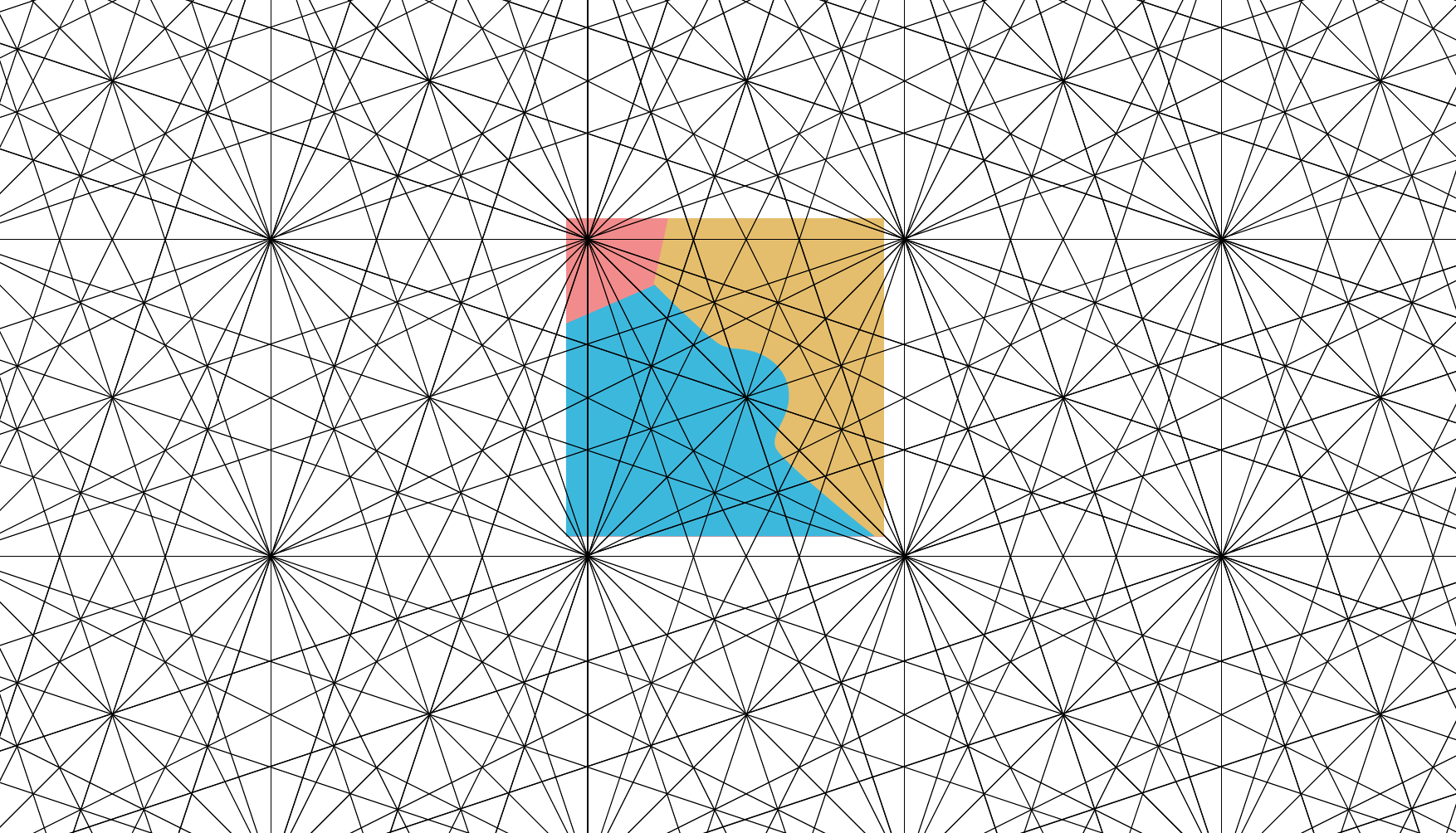}
    \caption{Illustration of the subdiagram considered in region $R_A$, further partitioned into three subdiagrams.}
    \label{fig:RA_zoom}
\end{figure}

The region $R_A$ is illustrated in red in Figure \ref{fig:RA}. If we zoom in on some portion of $R_A$, we get a repeating pattern illustrated in Figure \ref{fig:RA_zoom}. On this latter figure a subembedding $S_A$ is also illustrated by a unit square which is further divided into three subembeddings $S_A^1$ (red, top-left), $S_A^2$ (yellow, top-right) and $S_A^3$ (blue, bottom-left). 

The subembedding $S_A^1$ is an embedding of a $(1)_{12}$-matching, and thus has exactly $B_{12} = 2\ 894\ 710\ 651\ 370\ 536$ pseudochord arrangements. Our computations show that subembedding $S_A^2$ has $1\ 181\ 083\ 068$ pseudochord arrangements. The subembedding $S_A^3$ has $5\ 228\ 739\ 265\ 944$ pseudochord arrangements. By Proposition \ref{prop:sub_diagram}, the number of pseudochord arrangements for $S_A$ is at least
\begin{align*}
    n_A :&= 2\ 894\ 710\ 651\ 370\ 536\cdot 1\ 181\ 083\ 068 \cdot 5228739265944\\
    &= 17\ 876\ 503\ 929\ 228\ 145\ 018\ 796\ 772\ 391\ 568\ 838\ 912.
\end{align*}
Moreover, $S_A$ covers an area of $1$ (recall that there is a unit distance between two consecutive horizontal or vertical lines in the construction), while the region $R_A$ has an area of $m^2/12$. Thus, $S_A$ appears $p_A := m^2/12-O(m)$ times (independently) in $R_A$, taking border effects into account.

\paragraph*{Regions $R_B$, and $R_C$}

\begin{figure}[h]
    \centering\hfill
    \begin{minipage}{0.49\textwidth}
        \centering
        \includegraphics[width=0.67\textwidth]{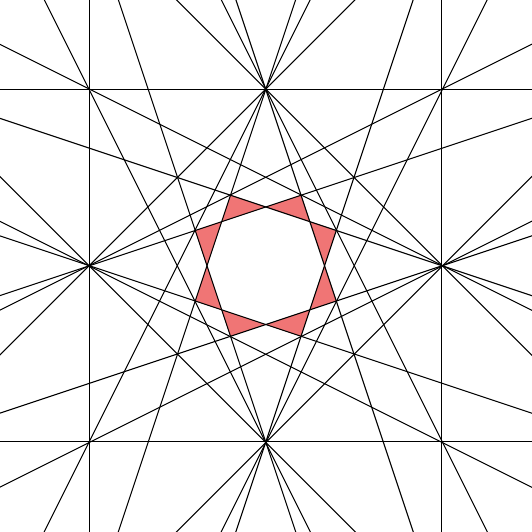}
        \caption{Region $R_B$, which has area $m^2/30$.}
        \label{fig:RB}
    \end{minipage}\hfill
    \begin{minipage}{0.49\textwidth}
        \centering
        \includegraphics[width=0.67\textwidth]{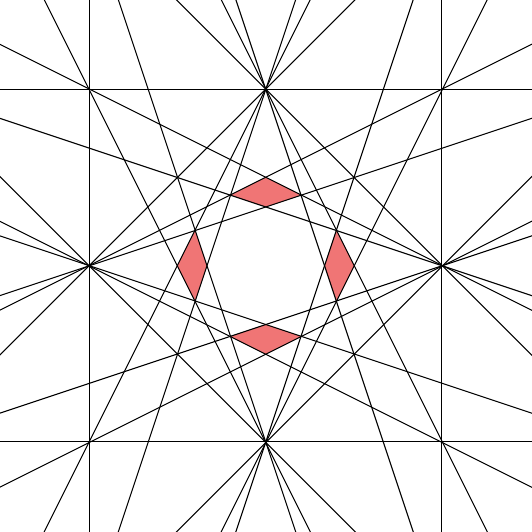}
        \captionsetup{textformat=simple}
        \caption{Region $R_C$, which has area $m^2/30$.}
        \label{fig:RC}
    \end{minipage}\hfill
\end{figure}

\begin{figure}[h]
    \centering\hfill
    \begin{minipage}{0.49\textwidth}
        \centering
        \includegraphics[trim={5cm 4cm 4.6cm 3cm}, clip, width=0.67\linewidth]{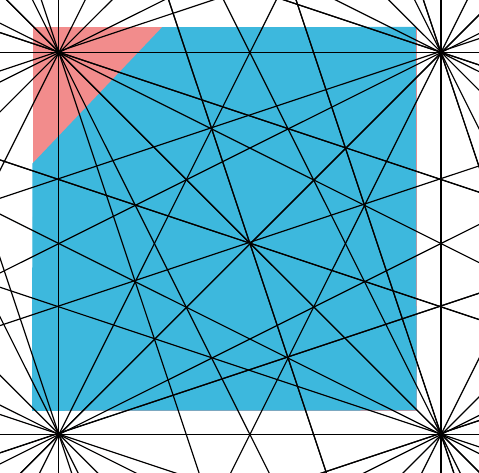}
        \caption{Subembeddings considered in $R_B$, further partitioned into two subembeddings.}
        \label{fig:RB_zoom}
    \end{minipage}\hfill
    \begin{minipage}{0.49\textwidth}
        \centering
        \includegraphics[trim={5cm 4cm 4.6cm 3cm}, clip, width=0.67\linewidth]{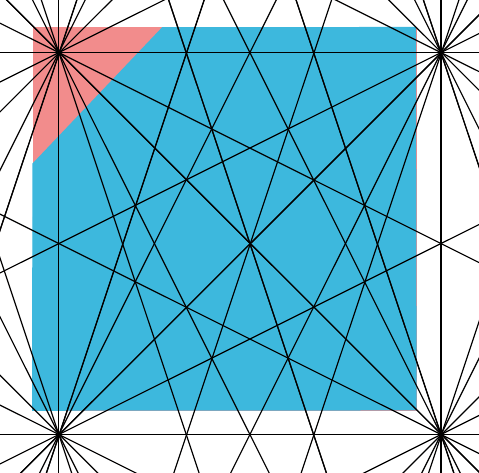}
        \captionsetup{textformat=simple}
        \caption{Subembedding considered in $R_C$, further partitioned into two subembeddings.}
        \label{fig:RC_zoom}
    \end{minipage}\hfill
\end{figure}

The regions $R_B$, and $R_C$ are illustrated in red in Figures \ref{fig:RB}, and \ref{fig:RC} respectively. If we zoom in on some portion of each of these regions, we get repeating patterns illustrated in Figures \ref{fig:RB_zoom} and \ref{fig:RC_zoom} respectively (up to symmetries). As previously, subembeddings $S_B$ and $S_C$ are also illustrated by unit squares further divided into two subembeddings $S_X^1$ (red, top-left) and $S_X^2$ (blue, bottom-right) for $X\in \{B,C\}$. 

The subembedding $S_B^1$ is an embedding of a $(1)_{11}$-matching, and thus has exactly $B_{11} = 5\ 449\ 192\ 389\ 984$ pseudochord arrangements. The subembedding $S_C^1$ is an embedding of a $(1)_{10}$-matching, and thus has exactly $B_{10} = 18410581880$ pseudochord arrangements. 

The subembedding $S_B^2$ and $S_C^2$ have $4\ 485\ 362\ 657\ 994\ 086$ and $6\ 674\ 057\ 692$ pseudochord arrangements respectively. By Proposition \ref{prop:sub_diagram}, the number of pseudochord arrangements for $S_B$ and $S_C$ are at least
\begin{align*}
    n_B &:= 5\ 449\ 192\ 389\ 984\cdot 4\ 485\ 362\ 657\ 994\ 086 = 24\ 441\ 604\ 062\ 259\ 780\ 293\ 677\ 634\ 624;\\
    n_C &:= 18\ 410\ 581\ 880\cdot 6\ 674\ 057\ 692 = 122\ 873\ 285\ 610\ 409\ 820\ 960.
\end{align*}

Moreover, both subembeddings cover an area of $1$, while the regions $R_B$ and $R_C$ both have area $m^2/30$. Thus, $S_B$ appears $p_B := m^2/30-O(m)$ times (independently) in $R_B$, and $S_C$ appears $p_C := m^2/30-O(m)$ times (independently) in $R_C$.

\paragraph*{Regions $R_D$ to $R_Q$}

\begin{figure}[!htb]
    \centering
    \begin{minipage}{0.28\textwidth}
        \centering
        \includegraphics[width=0.9\textwidth]{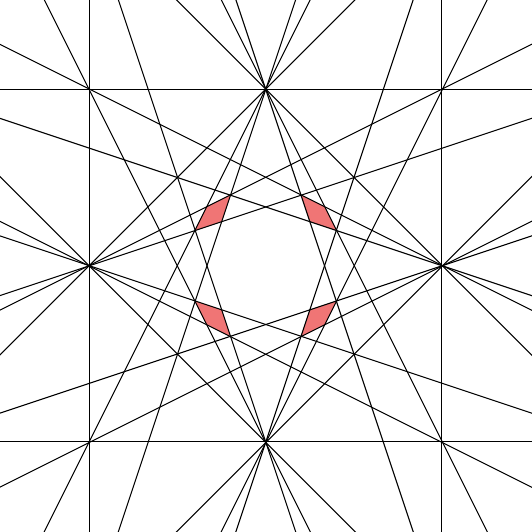}
        \subcaption{$R_D$, with area $m^2/60$.}
        \label{fig:RD}
    \end{minipage}\hfill
    \begin{minipage}{0.28\textwidth}
        \centering
        \includegraphics[width=0.9\textwidth]{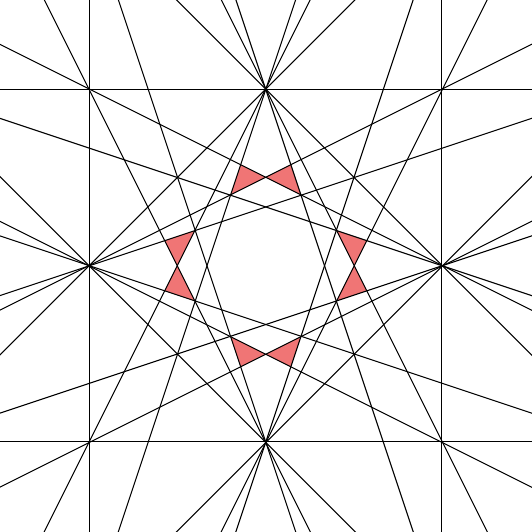}
        \captionsetup{textformat=simple}
        \subcaption{$R_E$, with area $m^2/35$.}
        \label{fig:RE}
    \end{minipage}\hfill
    \begin{minipage}{0.28\textwidth}
        \centering
        \includegraphics[width=0.9\textwidth]{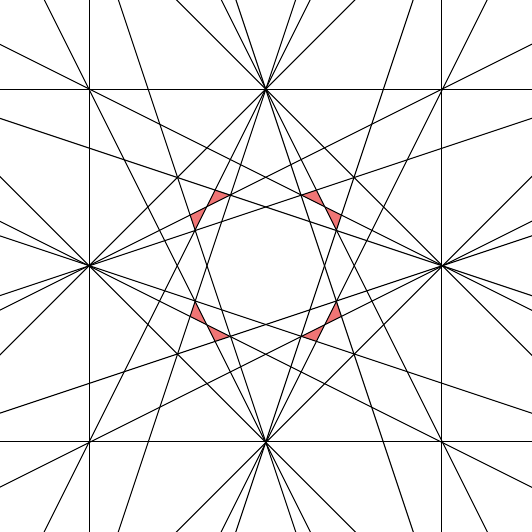}
        \captionsetup{textformat=simple}
        \subcaption{$R_F$, with area $m^2/105$.}
        \label{fig:RF}
    \end{minipage}\\

    \begin{minipage}{0.28\textwidth}
        \centering
        \includegraphics[width=0.9\textwidth]{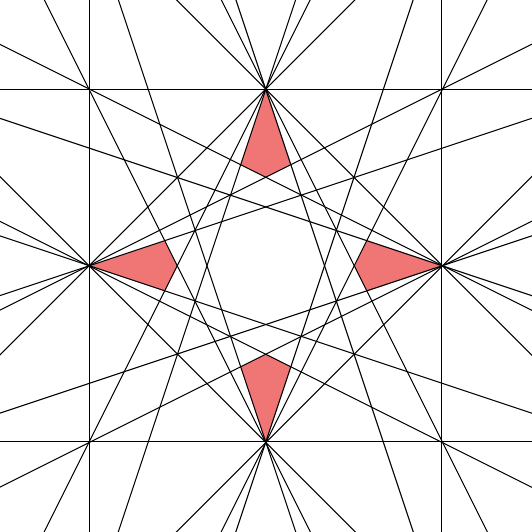}
        \subcaption{$R_G$, with area $m^2/14$.}
        \label{fig:RG}
    \end{minipage}\hfill
    \begin{minipage}{0.28\textwidth}
        \centering
        \includegraphics[width=0.9\textwidth]{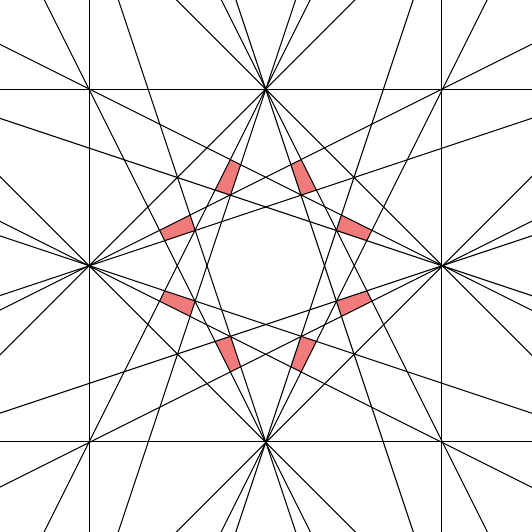}
        \captionsetup{textformat=simple}
        \subcaption{$R_H$, with area $m^2/35$.}
        \label{fig:RH}
    \end{minipage}\hfill
    \begin{minipage}{0.28\textwidth}
        \centering
        \includegraphics[width=0.9\textwidth]{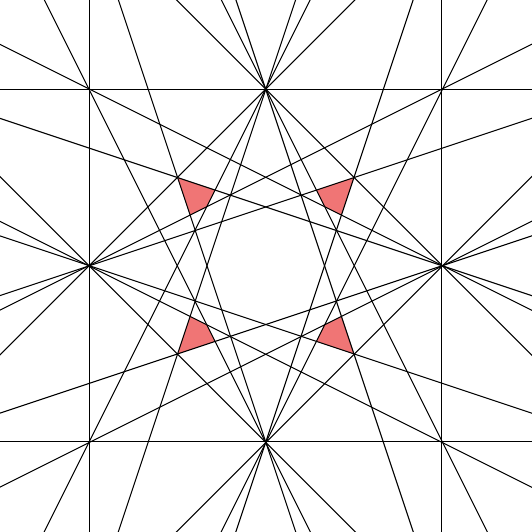}
        \captionsetup{textformat=simple}
        \subcaption{$R_I$, with area $m^2/42$.}
        \label{fig:RI}
    \end{minipage}\\

    \begin{minipage}{0.28\textwidth}
        \centering
        \includegraphics[width=0.9\textwidth]{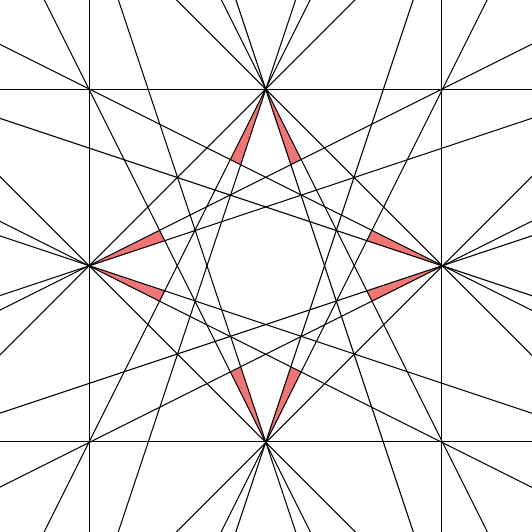}
        \subcaption{$R_J$, with area $m^2/35$.}
        \label{fig:RJ}
    \end{minipage}\hfill
    \begin{minipage}{0.28\textwidth}
        \centering
        \includegraphics[width=0.9\textwidth]{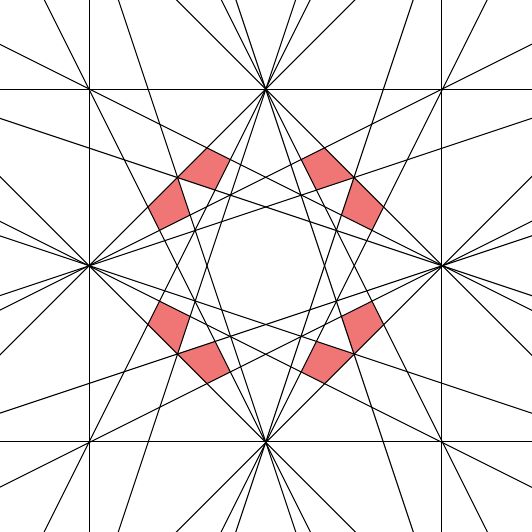}
        \captionsetup{textformat=simple}
        \subcaption{$R_K$, with area $8m^2/105$.}
        \label{fig:RK}
    \end{minipage}\hfill
    \begin{minipage}{0.28\textwidth}
        \centering
        \includegraphics[width=0.9\textwidth]{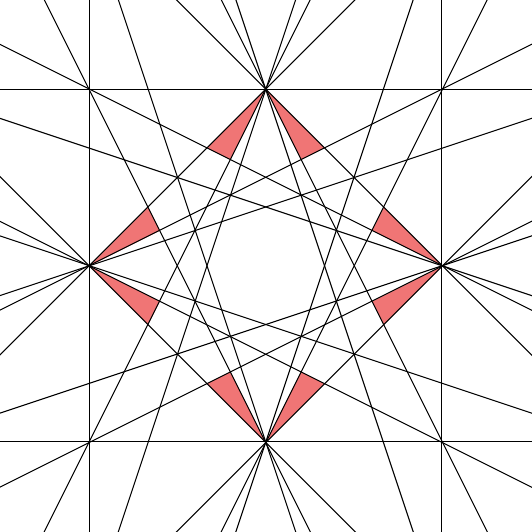}
        \captionsetup{textformat=simple}
        \subcaption{$R_L$, with area $m^2/15$.}
        \label{fig:RL}
    \end{minipage}\\

    \begin{minipage}{0.28\textwidth}
        \centering
        \includegraphics[width=0.9\textwidth]{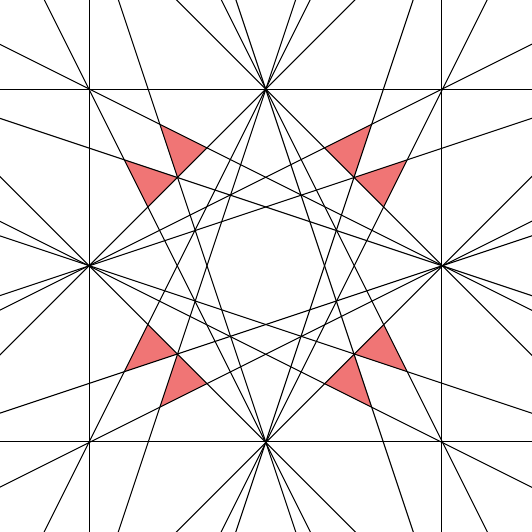}
        \subcaption{$R_M$, with area $m^2/15$.}
        \label{fig:RM}
    \end{minipage}\hfill
    \begin{minipage}{0.28\textwidth}
        \centering
        \includegraphics[width=0.9\textwidth]{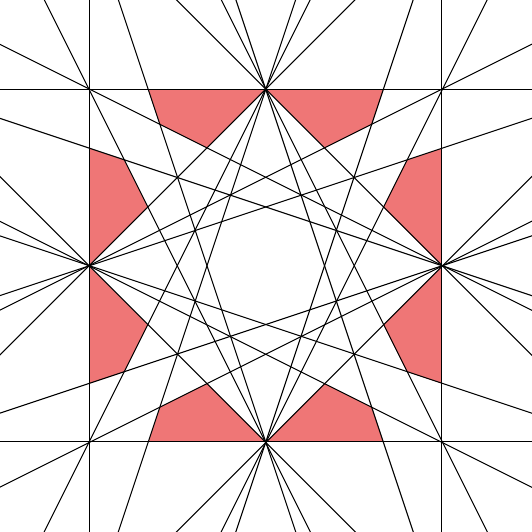}
        \captionsetup{textformat=simple}
        \subcaption{$R_N$, with area $4m^2/15$.}
        \label{fig:RN}
    \end{minipage}\hfill
    \begin{minipage}{0.28\textwidth}
        \centering
        \includegraphics[width=0.9\textwidth]{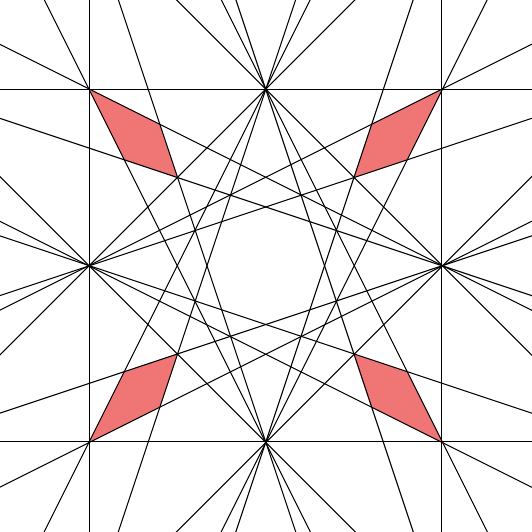}
        \captionsetup{textformat=simple}
        \subcaption{$R_O$, with area $m^2/10$.}
        \label{fig:RO}
    \end{minipage}\\

    \hfill
    \begin{minipage}{0.28\textwidth}
        \centering
        \includegraphics[width=0.9\textwidth]{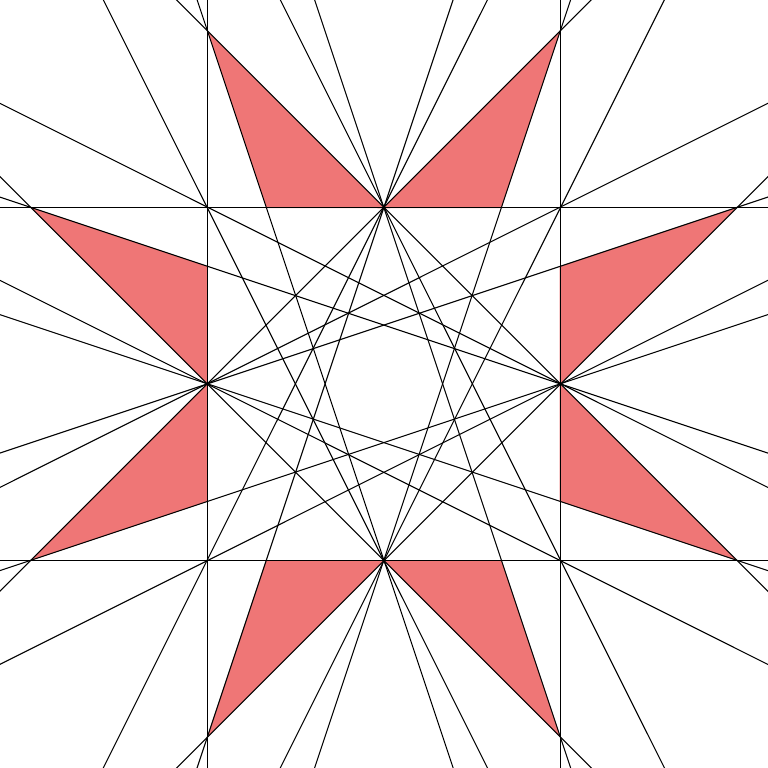}
        \subcaption{$R_P$, with area $2m^2/3$.}
        \label{fig:RP}
    \end{minipage}\hfill
    \begin{minipage}{0.28\textwidth}
        \centering
        \includegraphics[width=0.9\textwidth]{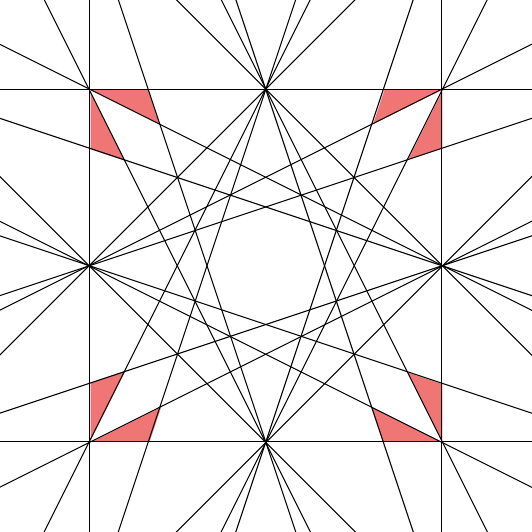}
        \captionsetup{textformat=simple}
        \subcaption{$R_Q$, with area $m^2/15$.}
        \label{fig:RQ}
    \end{minipage}\hfill\hfill

    \caption{Regions $R_D$ through $R_Q$.}\label{fig:RDQ}
\end{figure}

\begin{figure}[!htb]
    \centering
    \hfill
    \begin{minipage}{0.45\textwidth}
        \centering
        \includegraphics[trim={1cm 1cm 1cm 0cm}, clip, angle=90, width=0.9\linewidth]{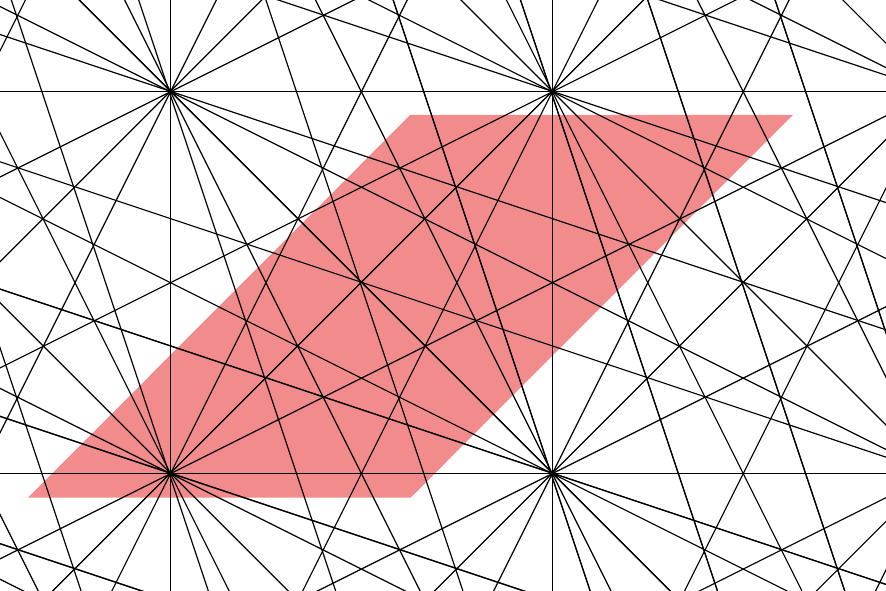}
        \subcaption{Subdiagram used for region $R_D$.}
        \label{fig:RD_zoom}
    \end{minipage} \hfill
    \begin{minipage}{0.45\textwidth}
        \centering
        \includegraphics[trim={4cm 4cm 4cm 3cm}, clip, width=0.7\linewidth]{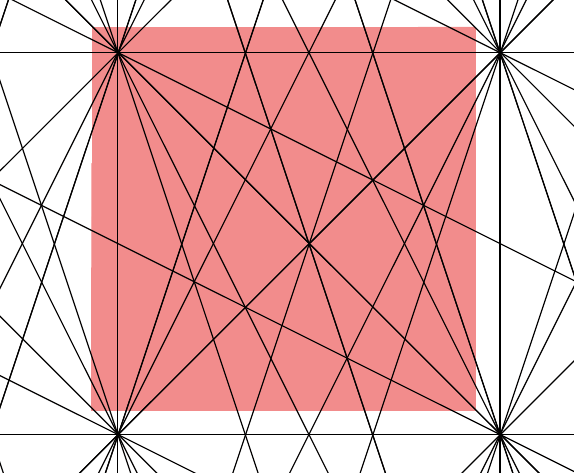}
        \captionsetup{textformat=simple}
        \subcaption{Subdiagram used for region $R_E$.}
        \label{fig:RE_zoom}
    \end{minipage} \hfill\hfill \\
    \hfill
    \begin{minipage}[b]{0.45\textwidth}
        \centering
        \includegraphics[trim={4cm 4cm 4cm 3cm}, clip, width=0.7\linewidth]{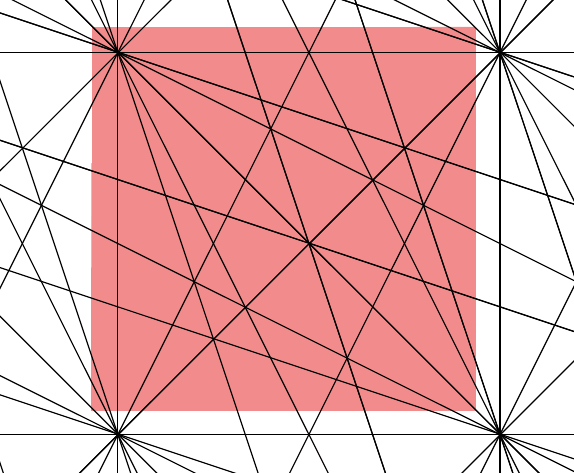}
        \subcaption{Subdiagram used for region $R_F$.}
        \label{fig:RF_zoom}
    \end{minipage}
    \hfill
    \begin{minipage}[b]{0.45\textwidth}
        \centering
        \includegraphics[trim={0cm 1cm 1cm 1cm}, clip, angle=90, width=0.9\linewidth]{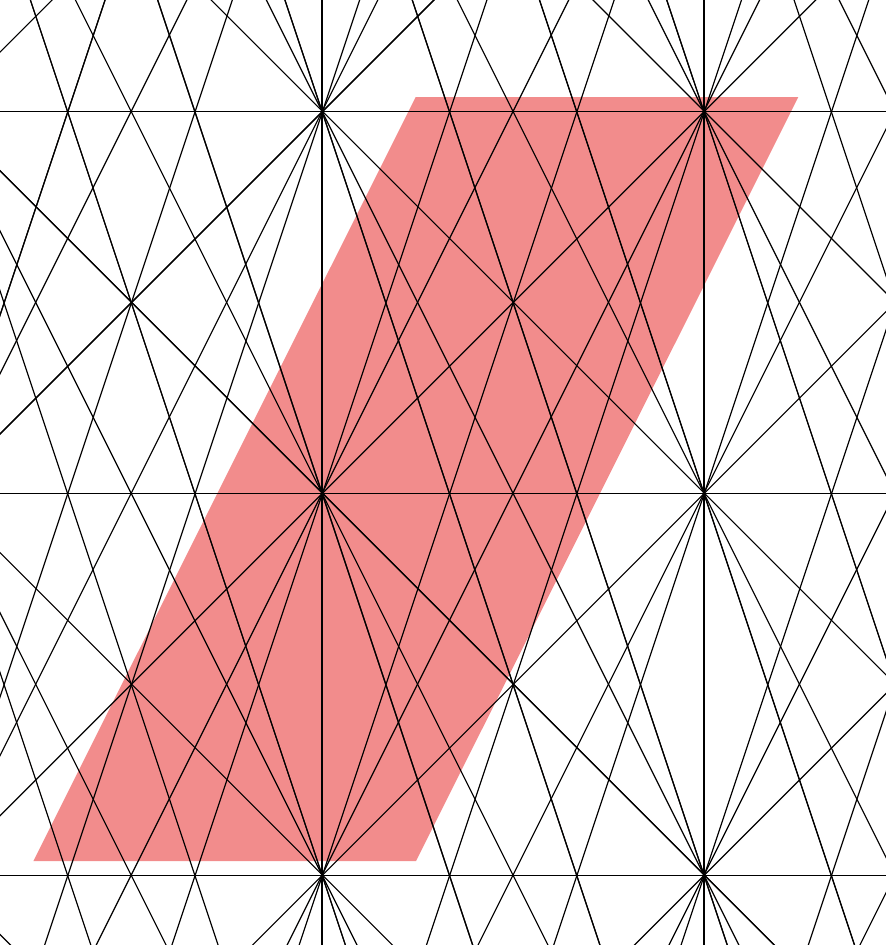}
        \captionsetup{textformat=simple}
        \subcaption{Subdiagram used for region $R_G$.}
        \label{fig:RG_zoom}
    \end{minipage} \hfill

    \caption{Subdiagrams used in regions $R_D$ through $R_G$. The subdiagram used for $R_G$ has area $2$, while the three others have area $1$.}\label{fig:RDG_zoom}
    
\end{figure}

\begin{table}[!htb]
    \centering
    \begin{tabular}{c|c|c}
        Region & $n_{\bullet}$ & $p_\bullet$ \\\hline
        $R_A$ & $17\ 876\ 503\ 929\ 228\ 145\ 018\ 796\ 772\ 391\ 568\ 838\ 912$ & $m^2/12-O(m)$ \\
        $R_B$ &  $24\ 441\ 604\ 062\ 259\ 780\ 293\ 677\ 634\ 624$ & $m^2/30-O(m)$ \\
        $R_C$ &  $122\ 873\ 285\ 610\ 409\ 820\ 960$ & $m^2/30-O(m)$ \\
        $R_D$ &  $145\ 267\ 240\ 140\ 131\ 510\ 094$ & $m^2/60-O(m)$ \\
        $R_E$ &  $884\ 854\ 135\ 426\ 438$ & $m^2/35-O(m)$ \\
        $R_F$ &  $4\ 354\ 539\ 523\ 065\ 118$ & $m^2/105-O(m)$ \\
        $R_G$ &  $134\ 841\ 117\ 561\ 581\ 177\ 808$ & $m^2/28-O(m)$ \\
        $R_H$ &  $21\ 027\ 918\ 182$ & $m^2/35-O(m)$ \\
        $R_I$ &  $1\ 422\ 375\ 838\ 634\ 144\ 387\ 571$ & $m^2/84-O(m)$ \\
        $R_J$ &  $36\ 797\ 080\ 857\ 271\ 908\ 723$ & $m^2/105-O(m)$ \\
        $R_K$ &  $42\ 961\ 411\ 048\ 824$ & $m^2/210-O(m)$ \\
        $R_L$ &  $23\ 454\ 005\ 259\ 745\ 292$ & $m^2/60-O(m)$ \\
        $R_M$ &  $15\ 342\ 798\ 480\ 294\ 823$ & $m^2/60-O(m)$ \\
        $R_N$ &  $50\ 236\ 135\ 250\ 760$ & $2m^2/45-O(m)$ \\
        $R_O$ &  $50\ 236\ 135\ 250\ 760$ & $m^2/15-O(m)$ \\
        $R_P$ &  $104\ 878\ 461\ 268\ 633\ 368\ 974\ 367$ & $2m^2/63-O(m)$ \\
        $R_Q$ &  $104\ 878\ 461\ 268\ 633\ 368\ 974\ 367$ & $m^2/315-O(m)$ \\
        $R_R$ &  $>2^{349\, 033}$ & $m^2/500^2-O(m)$ \\
        $R_S$ &  $>2^{349\, 033}$ & $m^2/(3\cdot 500^2)-O(m)$ \\
    \end{tabular}
    \caption{The number of arrangements for the chosen subembedding for each region ($n_\bullet$) and the number of times said subembedding appears in the corresponding region ($p_\bullet$).}
    \label{tab:counts}
\end{table}

The regions $R_D$ through $R_Q$ are illustrated in Figure \ref{fig:RDQ}.
Here, we can afford to compute the number of arrangements for a single subembedding each (without further breaking it down into further subembeddings), or even consider larger subembeddings, depending on the region. The chosen subembeddings for regions $R_D$ through $R_G$, are illustrated in Figure \ref{fig:RDG_zoom} (the subembeddings chosen for regions $R_H$ to $R_Q$ can be found in the appendix). The number of simple arrangements for each subembedding, and the number of times each subembedding appears in its respective region, are reported in Table \ref{tab:counts}.

\paragraph*{Regions $R_R$ and $R_S$}
\begin{figure}[!h]
    \centering
    \hfill
    \begin{minipage}{0.47\textwidth}
        \centering
        \includegraphics[width=0.8\textwidth]{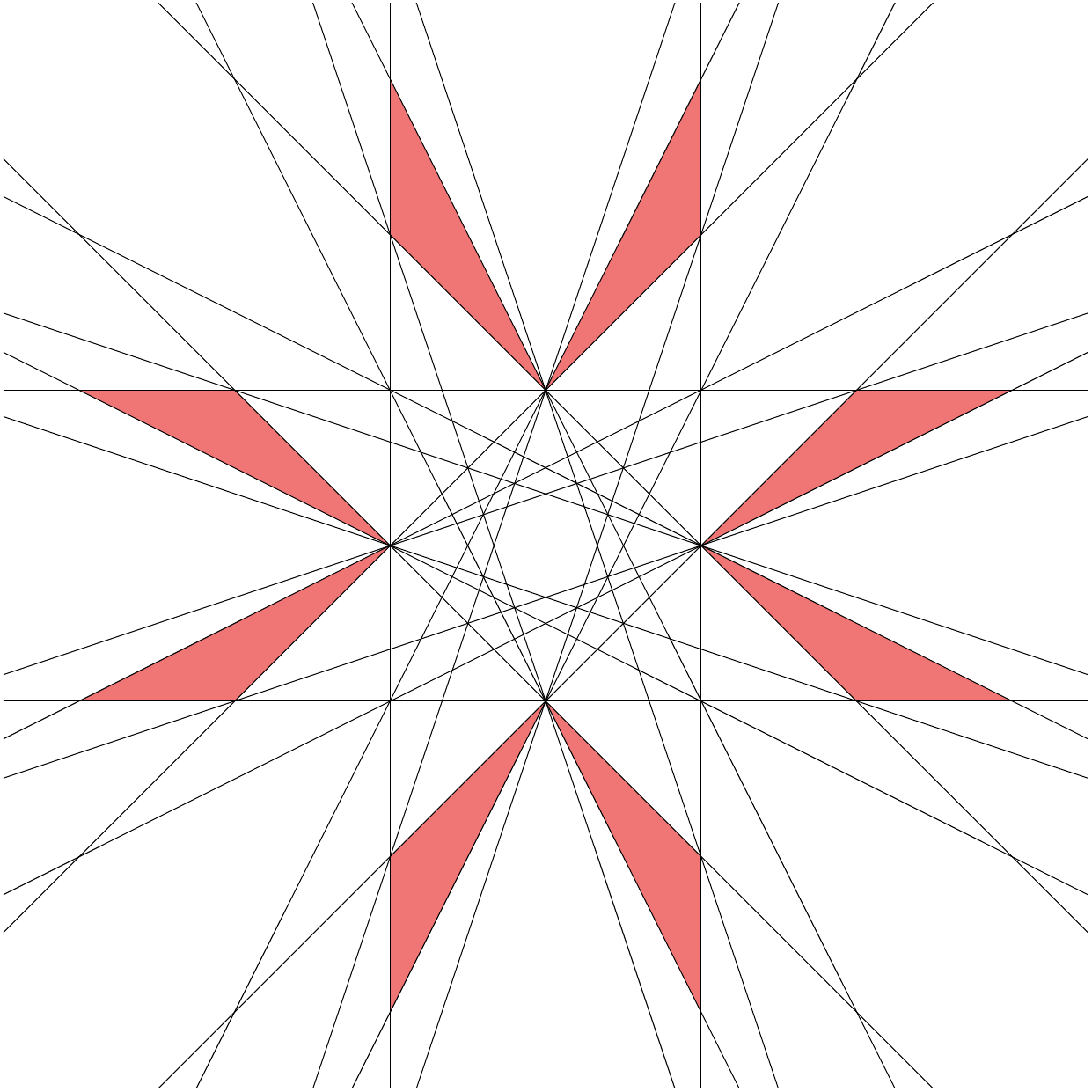}
        \caption{Region $R_R$, which has area $m^2$.}
        \label{fig:RR}
    \end{minipage}\hfill
    \begin{minipage}{0.47\textwidth}
        \centering
        \includegraphics[width=0.8\textwidth]{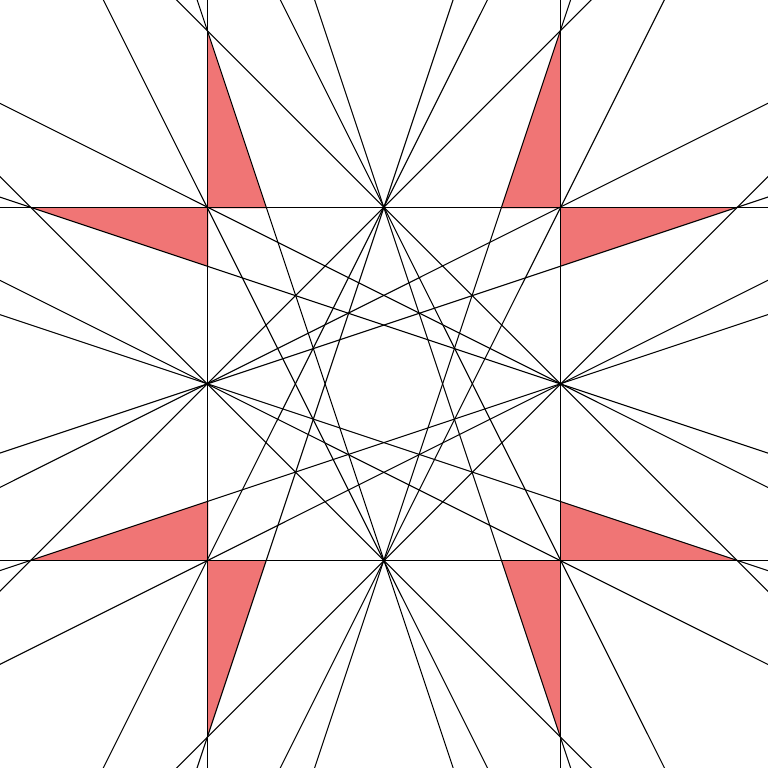}
        \captionsetup{textformat=simple}
        \caption{Region $R_S$, which has area $m^2/3$.}
        \label{fig:RS}
    \end{minipage}\hfill\hfill
\end{figure}

\begin{figure}
    \centering
    \hfill
    \begin{minipage}{0.32\textwidth}
        \centering
        \includegraphics[width=0.9\textwidth]{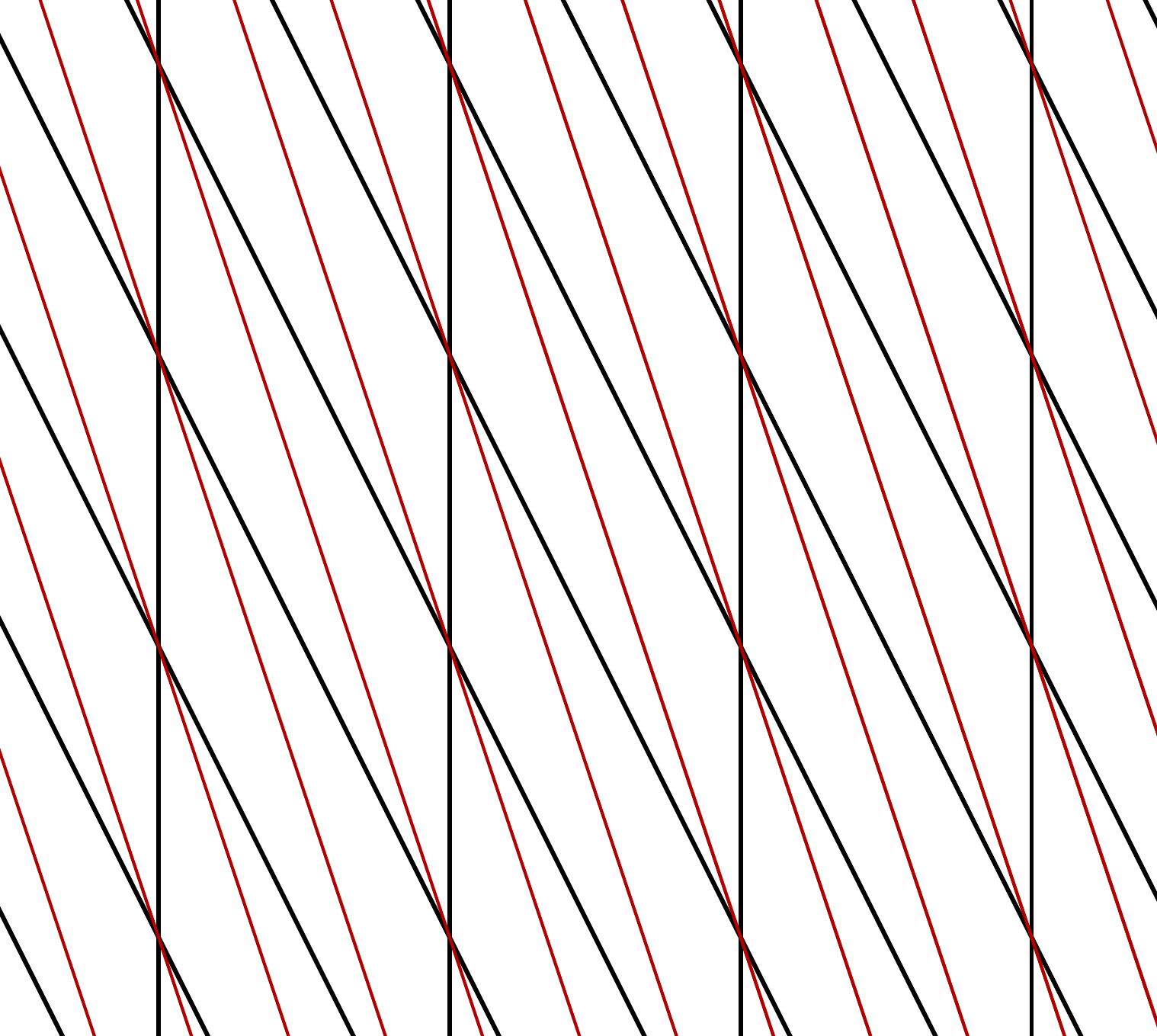}
        \subcaption{Pattern inside region $R_R$.}
        \label{fig:RR_zoom}
    \end{minipage}\hfill
    \begin{minipage}{0.32\textwidth}
        \centering
        \includegraphics[width=0.9\textwidth]{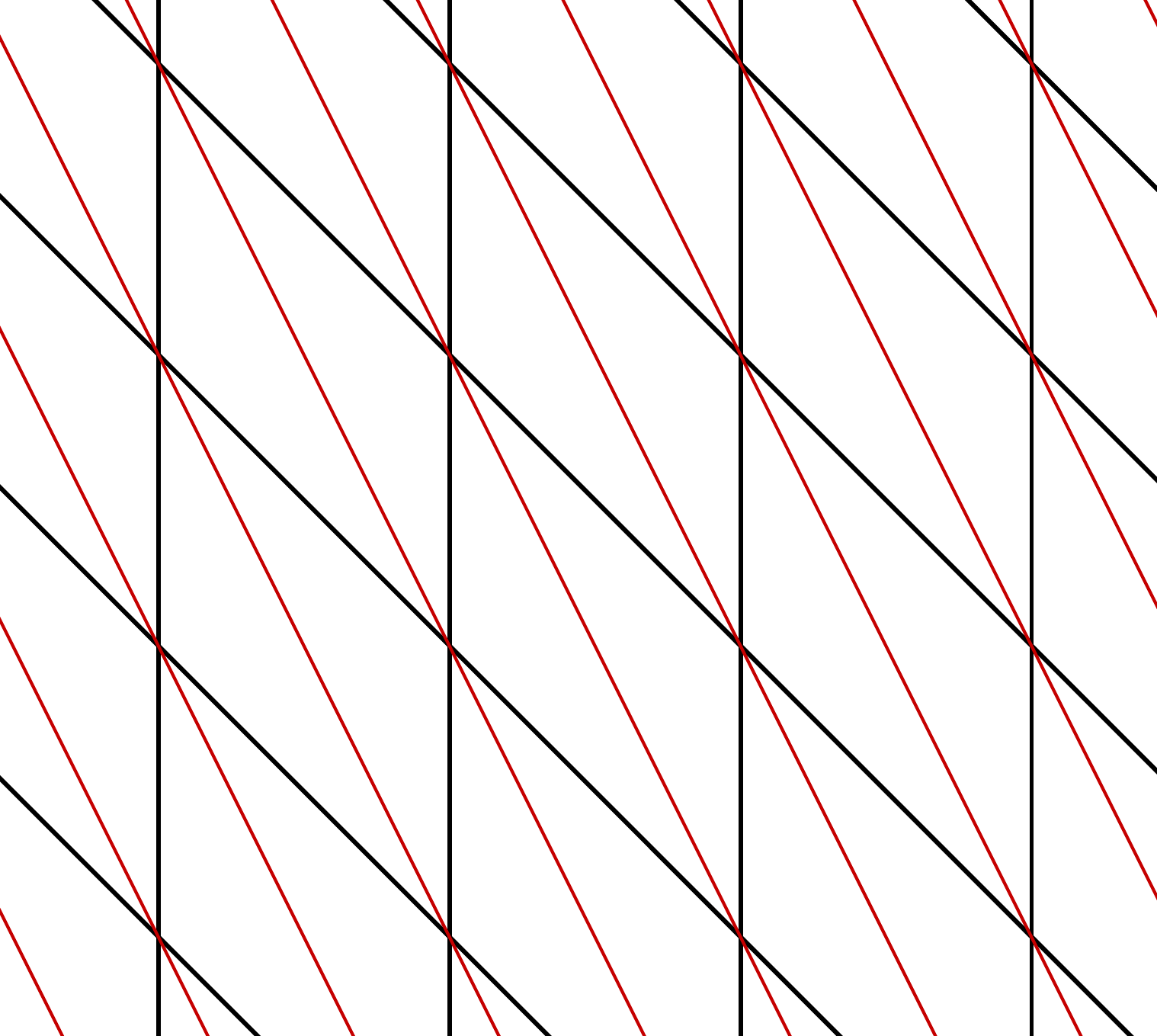}
        \captionsetup{textformat=simple}
        \subcaption{Pattern inside region $R_S$.}
        \label{fig:RS_zoom}
    \end{minipage}\hfill\hfill
    \hfill
    \begin{minipage}{0.32\textwidth}
        \centering
        \includegraphics[width=0.9\textwidth]{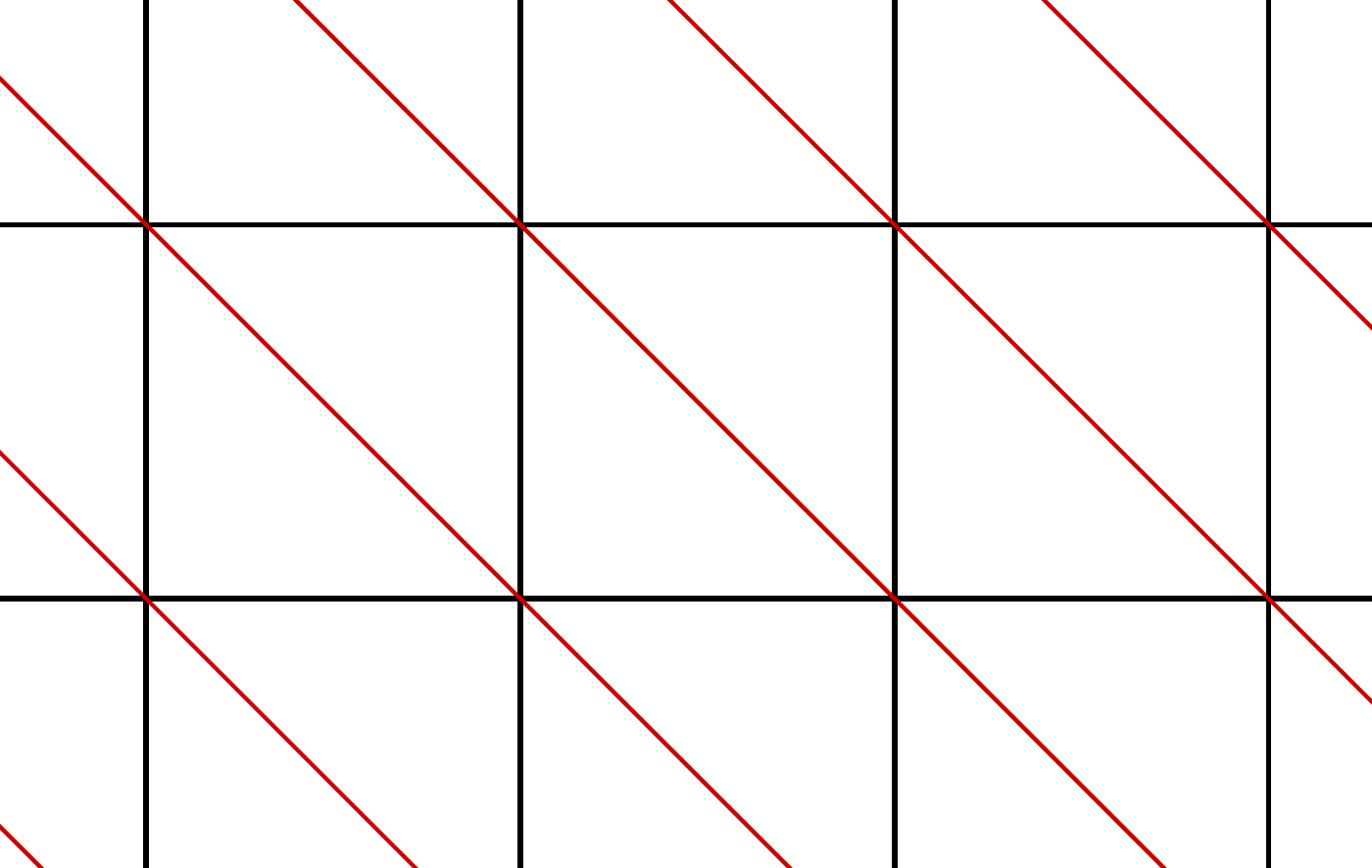}
        \subcaption{Pattern obtained from both after vertical shear mappings.}
        \label{fig:R_shear_zoom}
    \end{minipage}\hfill\hfill

    \caption{Illustration of the fact that up to a vertical shear mapping (which preserves areas), the patterns inside regions $R_R$ and $R_S$ can be viewed as the same regular grid with additional lines of slope $-1$.}
    \label{fig:RRS_shear}
\end{figure}

Up to area preserving linear transformations, the patterns of lines inside regions $R_R$ and $R_S$ are the same, as illustrated in Figure \ref{fig:RRS_shear}. We will thus consider the pattern illustrated on the right of this figure. While for all other regions, we used the approach explained in Section \ref{section:counting} to count the number of arrangements of a chosen subembedding, the specific pattern in these two regions allows for a different, much more efficient approach. 

\begin{figure}
    \centering
    \includegraphics[width=0.75\linewidth]{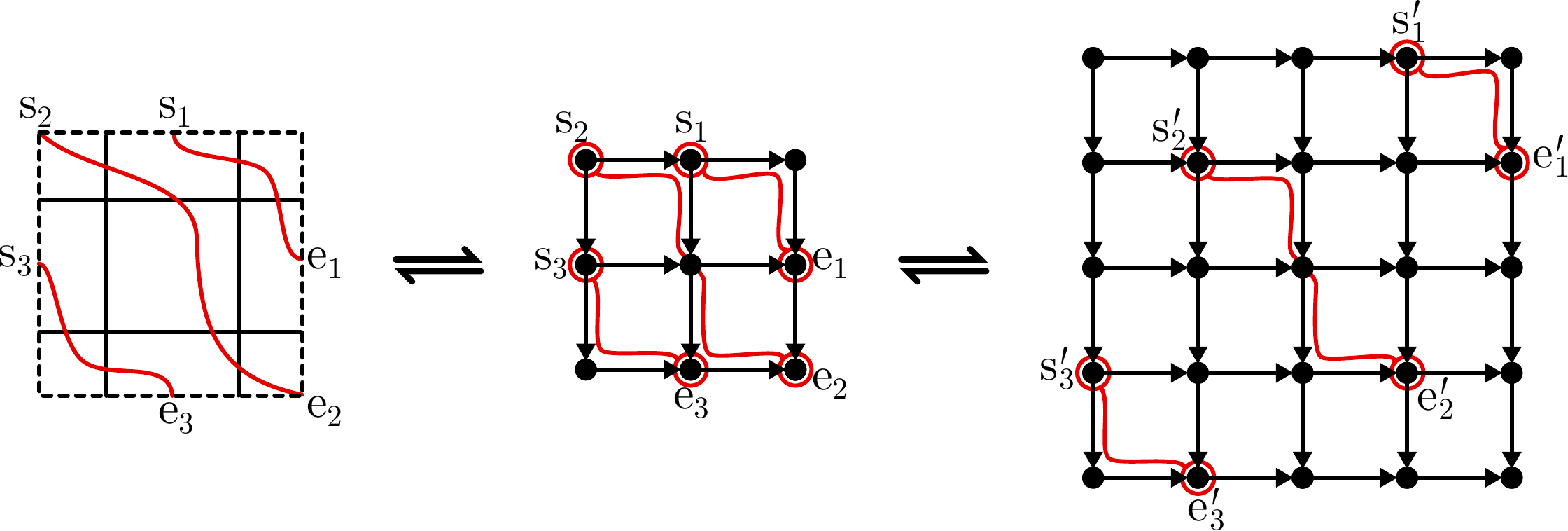}
    \caption{Illustration of the correspondence between arrangements of the chord diagram and disjoint paths in the DAG.}
    \label{fig:gessel-viennot}
\end{figure}

We illustrate it on a small subembedding defined by a square of side length $2$ (see Figure \ref{fig:gessel-viennot}). We can count the number of arrangements of this subembedding by starting with only the black grid, and counting the number of combinatorially distinct ways to insert three curves starting at $s_1$, $s_2$ and $s_3$ respectively, and ending at $e_1$, $e_2$ $e_3$ respectively, such that these three curves do not intersect, and each crosses every black segment at most once. We can in turn translate this into a question about counting the number of ways to have three paths in a certain directed acyclic grid graph, starting and ending at prescribed vertices, such that these paths do not cross (middle illustration in Figure \ref{fig:gessel-viennot}). Note that here by ``not crossing'' we do not mean vertex- or edge-disjoint, but that if a path starts above/right of another, it can at no point go below/left of it. Given three such paths, ordered from top right to bottom left, we can shift the second path one unit down and left along the grid and the second path two units down and left along the grid, we get a set of three vertex-disjoint paths (right illustration in Figure \ref{fig:gessel-viennot}). The reverse is also true: given three vertex-disjoint paths on the directed acyclic graph starting at ${s'_1, s'_2, s'_3}$ and ending at ${e'_1, e'_2, e'_3}$, reversing the shifts produces three non-crossing paths starting at ${s_1, s_2, s_3}$ and ending at ${e_1, e_2, e_3}$. 

Thus, our question finally becomes counting the number of ways to have three such vertex-disjoint paths, whose set of starting vertices is ${s'_1, s'_2, s'_3}$ and set of ending vertices is ${e'_1, e'_2, e'_3}$ (note that the vertex-disjointness ensures that in such a set of paths, the path starting at $s'_i$ will end at $e'_i$ for all $1\leq i \leq 3$). The Lindström–Gessel–Viennot lemma (or a weakened version of it) gives us an efficient method to do so:
\begin{lemma}[Lindström \cite{Linstrom1973}, Gessel \& Viennot \cite{Gessel1985}]
    Let $G$ be a finite directed acyclic graph. Consider starting vertices $S=\{s_{1},\ldots ,s_{k}\}$ and destination vertices $E=\{e_{1},\ldots ,e_{k}\}$. For any two vertices $u$ and $v$, let $p(u,v)$ be the number of paths from $u$ to $v$. Assume that for any tuple of $k$ vertex-disjoint paths starting in $S$ and ending in $E$, the path starting at $s_i$ necessarily ends at $e_i$, for all $1\leq i \leq k$. Then the number of distinct such tuples is the determinant of the following matrix:
    \[M =\begin{pmatrix}
            p(s_1,e_1) & p(s_1, e_2) & \ldots & p(s_1, e_k)\\
            p(s_2,e_1) & p(s_2, e_2) & \ldots & p(s_2, e_k)\\
            \vdots & \vdots & \vdots & \vdots\\
            p(s_k,e_1) & p(s_k, e_2) & \ldots & p(s_k, e_k)\\
        \end{pmatrix}\]
\end{lemma}

In our specific case, the entries of this matrix are the number of paths going only down or right between two specified vertices in a grid, which can be easily expressed as a binomial coefficient. We have
\[M =   \begin{pmatrix}
            \binom{2}{1} & \binom{3}{0} &  0\\
            \binom{3}{3} & \binom{4}{2} & \binom{3}{0}\\
            0 & \binom{3}{3} & \binom{2}{1}\\
        \end{pmatrix}
    =   \begin{pmatrix}
            2 & 1 &  0\\
            1 & 6 & 1\\
            0 & 1 & 2\\
        \end{pmatrix}.
        \]
The determinant of this matrix is $|M|=20$, and we recover the count which was already used in the warm-up and illustrated in Figure \ref{fig:matousek-diagrams}. In general, for a square subembedding of size $s\times s$, the corresponding matrix is $M = (m_{ij})_{1\leq i,j \leq 2s-1}$, with
\[m_{ij} = \binom{2s-|s-i|-|s-j|}{\frac{2s-|s-i|-|s-j|+3|i-j|}{2}}.\]
For our actual bound, we use a square subembedding of side length $s=500$. The number of simple arrangements $n_R$ ($=n_S$) of this subembedding is too big to write explicitly here (in base $10$, it is $105\ 070$ digits long). But it is enough for us to know that $\log_2{n_R} > 349\ 033$, and that the number of times this subembedding appears independently in regions $R_R$ and $R_S$ is, respectively, $p_R := m^2/500^2 - O(m)$ and $p_S := m^2/(3\cdot500^2) - O(m)$.

\paragraph*{Putting everything together}

We can now put everything together by applying Proposition \ref{prop:sub_diagram}. This gives us the following lower bound for the number of simple pseudochord arrangements of a $(12)_m$-matching (remember that we started with an embedding of such a matching):
\[\prod_{X\in\{A,B,\ldots S\}} (n_X)^{p_X} = 2^{\left(\sum_{X} p_X \log_2{n_X}\right)}.\]
Computing this bound yields $2^{c\cdot m^2 - O(m)}$ where $c > 36.65$. Thus we have the following:
\begin{theorem}
For large enough $m$, the number of simple pseudochord arrangements of a $(12)_m$-matching is greater than $2^{c\cdot m^2}$, for some constant $c> 34.374$.
\end{theorem}

By now applying Proposition \ref{prop:recursion} we get our main result as a corollary:
\begin{corollary}
For large enough $n$, the number of simple line arrangements of order $n$ is greater than $2^{c \cdot m^2}$, for some constant $c>0.2604$.
\end{corollary}

\section{Implementation and computation time}

We ran all the the code on a 2020 MacBook Pro with an M1 chip. 
The code for computing the number of pseudochord arrangement for a given matching was written in Rust and can be found on GitHub~\cite{code}. It was parallelized to make use of all $8$ cores available on the machine (the algorithm is easy to parallelize, and for large enough chord matchings the speedup is close to linear in the number of cores used). The total CPU time (summed over the $8$ cores) used to compute the number of arrangements for all matchings used in our construction (excluding the $(1)_{12}$, $(1)_{11}$, and $(1)_{10}$ matchings, for which the numbers of pseudochord arrangements were previously known) was approximately $85$ hours.

\section{Discussion}

Many choices in the present construction were made heuristically or due to constraints in computing power and time. There are thus many avenues for possible improvements of the bound we obtain. We list a few here.
\begin{itemize}
    \item There is a lot of freedom in the choice of the initial embedding of the matching. Here we have worked with Dumitrescu and Mandal's ``rectangular construction with 12 slopes'' but it is not unlikely that a different matching or embedding could yield better bounds. The best bound obtained by these authors is based on a different construction (the so-called ``hexagonal construction with 12 slopes''), but our experiments applying our method on this construction have not yielded better bounds than with the rectangular construction. One can also try more different slopes. In particular, one could try to exploit the recent work of Rote \cite{Rote2023}, where he counts the number of pseudoline arrangements of order $16$ (although once again, we were unable to do so fruitfully). Note also that the initial embedding need not have chords which are straight line segments.
    \item The choices of subembeddings we made in each region are largely arbitrary. We have tried many different choices, but have no guarantee that the ones we made were good.
    \item Continuing on the previous point, being able to count the number of pseudochord arrangements faster or for larger subembeddings should lead to improvements. This could be achieved by better algorithms, better implementations, or simply more computing power. In particular, as the algorithm we use is highly parallelizable, running it on a computing cluster with many cores is an obvious avenue for improvements. One could also try to find and use faster ad-hoc algorithms for different types of matchings.
\end{itemize}

\bibstyle{plainurl}
\bibliography{refs}

\newpage
\section*{Appendix : Subembeddings chosen for regions $R_H$ to $R_Q$.}
\begin{figure}[!h]
    \centering
    \includegraphics[scale=0.5]{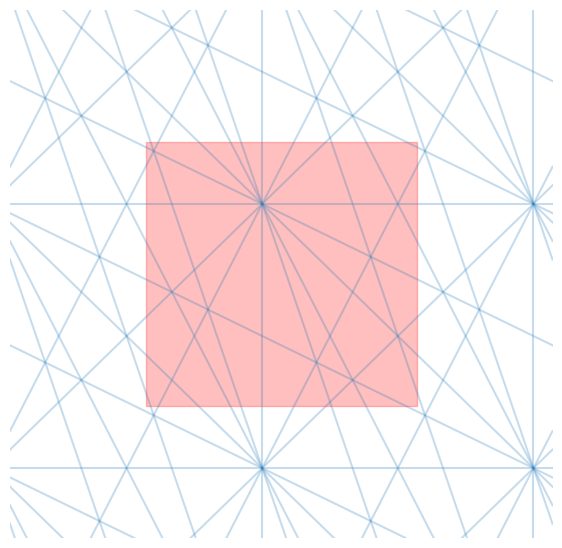}
    \caption{Subembedding chosen for region $R_H$.}
\end{figure}
\begin{figure}[!h]
    \centering
    \includegraphics[scale=0.7]{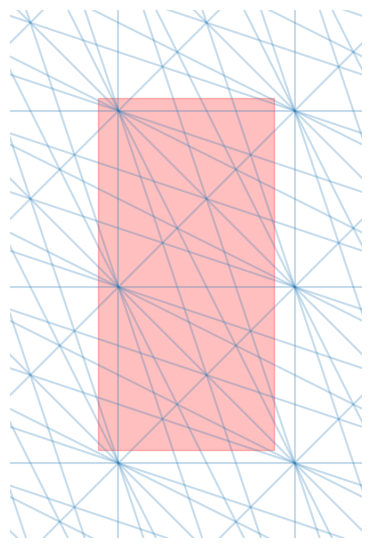}
    \caption{Subembedding chosen for region $R_I$.}
\end{figure}
\begin{figure}[!h]
    \centering
    \includegraphics[scale=0.8]{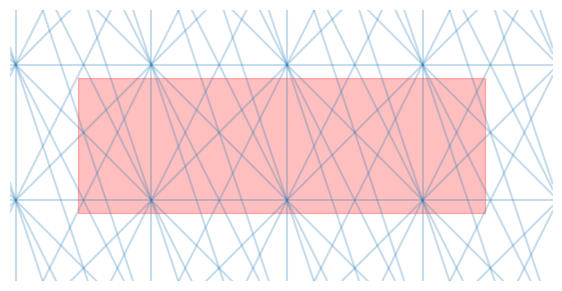}
    \caption{Subembedding chosen for region $R_J$.}
\end{figure}
\begin{figure}[!h]
    \centering
    \includegraphics[scale=0.8]{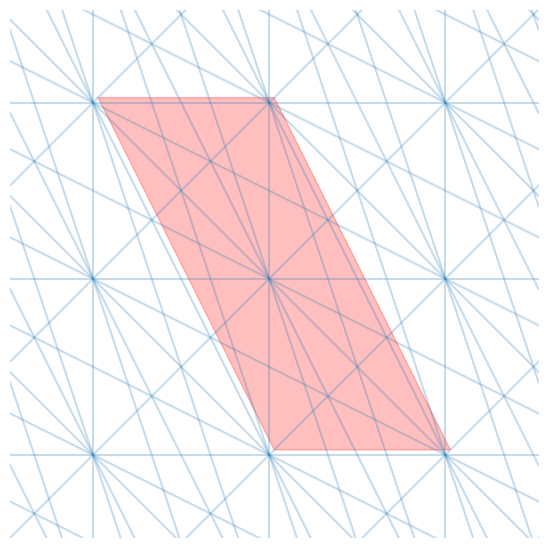}
    \caption{Subembedding chosen for region $R_K$.}
\end{figure}
\begin{figure}[!h]
    \centering
    \includegraphics[scale=0.6]{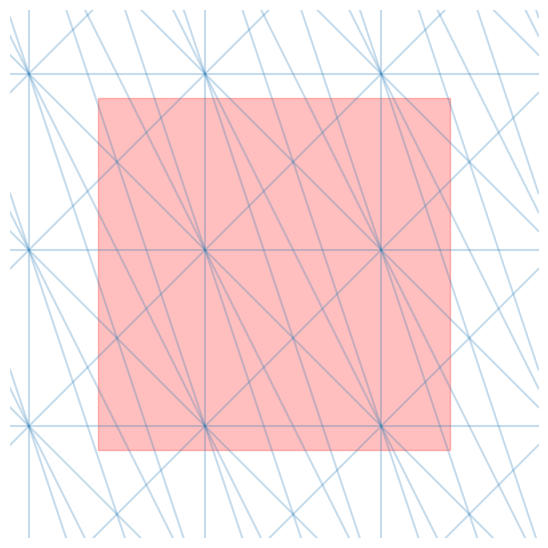}
    \caption{Subembedding chosen for region $R_L$.}
\end{figure}
\begin{figure}[!h]
    \centering
    \includegraphics[scale=0.6]{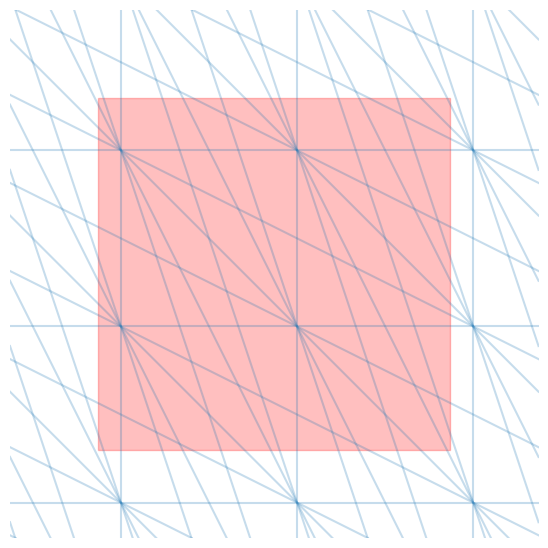}
    \caption{Subembedding chosen for region $R_M$.}
\end{figure}
\begin{figure}[!h]
    \centering
    \includegraphics[scale=0.6]{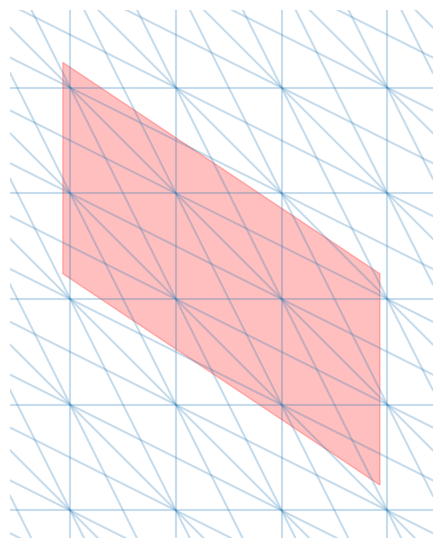}
    \caption{Subembedding chosen for regions $R_N$ and $R_O$ (up to an area preserving linear transformation, the patterns inside these two regions are identical).}
\end{figure}
\begin{figure}[!h]
    \centering
    \includegraphics[scale=0.8]{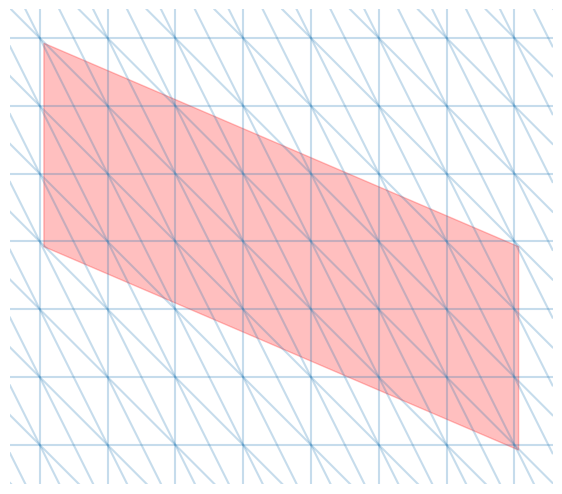}
    \caption{Subembedding chosen for regions $R_P$ and $R_Q$ (up to an area preserving linear transformation, the patterns inside these two regions are identical).}
\end{figure}
\end{document}